\documentclass[11pt]{amsart}
\usepackage{packages}

\title[Generalized cuts and Grothendieck covers]{%
  Generalized Cuts and Grothendieck Covers:\\
  a Primal-Dual Approximation Framework\\
  extending the Goemans--Williamson Algorithm
}
\author[N. Benedetto Proença]{Nathan Benedetto Proença\textsuperscript{1\P}}
\address{%
  \textsuperscript{1}Department of Combinatorics and Optimization, University of Waterloo
}
\thanks{%
  \textsuperscript{\P}%
  Research of this author was supported in part by a
  Discovery Grant from the Natural Sciences and Engineering
  Research Council (NSERC) of Canada
}
\author[M.K. de Carli Silva]{Marcel K. de Carli Silva\textsuperscript{2\textasteriskcentered}}
\address{%
  \textsuperscript{2}Institute of Mathematics and Statistics, University of São Paulo%
}
\thanks{%
  \textsuperscript{\textasteriskcentered}%
  This work was partially supported by Conselho Nacional de
  Desenvolvimento Científico e Tecnológico (CNPq).%
}
\author[C.M. Sato]{Cristiane M. Sato\textsuperscript{3\textasteriskcentered}}
\address{%
  \textsuperscript{3}Center for Mathematics, Computing and Cognition, Federal University of the ABC Region%
}
\author[L. Tunçel]{Levent Tunçel\textsuperscript{1\P}}

\begin{document}

\begin{abstract}
  We provide a primal-dual framework for randomized approximation
  algorithms utilizing semidefinite programming (SDP) relaxations.
  Our framework pairs a continuum of APX-complete problems including
  MaxCut, Max2Sat, MaxDicut, and more generally, Max-Boolean
  Constraint Satisfaction and MaxQ (maximization of a positive
  semidefinite quadratic form over the hypercube) with new
  APX-complete problems which are stated as convex optimization
  problems with exponentially many variables.
  These new dual counterparts, based on what we call
  \emph{Grothendieck covers}, range from fractional cut covering
  problems (for MaxCut) to tensor sign covering problems (for~MaxQ).
  For each of these problem pairs, our framework transforms the
  randomized approximation algorithms with the best known
  approximation factors for the primal problems to randomized
  approximation algorithms for their dual counterparts with reciprocal
  approximation factors which are tight with respect to the Unique
  Games Conjecture.
  For each APX-complete pair, our algorithms solve a single SDP
  relaxation and generate feasible solutions for both problems which
  also provide approximate optimality certificates for each other.
  Our work utilizes techniques from areas of randomized approximation
  algorithms, convex optimization, spectral sparsification, as well as
  Chernoff-type concentration results for random matrices.
\end{abstract}

\maketitle

\newlength{\eqboxwidth}
\settowidth{\eqboxwidth}{\(\coloneqq\)}
\newcommand{\eqaligned}{\mathrel{\makebox[\eqboxwidth][r]{=}}}

\thispagestyle{empty}
\setcounter{page}{0}

\clearpage

\section{Introduction}

Some of the most impressive successes for randomized approximation
algorithms, utilizing semidefinite programming relaxations, have been
on problems such as MaxCut~\cite{GoemansWilliamson1995},
Max2Sat~\cite{LewinLivnatZwick2002}, and
MaxDicut~\cite{BrakensiekHuangPotechinZwick2023}.
We define APX-complete duals for such problems, which involve what we
call \emph{Grothendieck covers}.
Then, we design a primal-dual framework of randomized approximation
algorithms for a wide range of problems, including maximum Boolean
constraint satisfaction problems (CSPs) paired with their APX-complete
duals, which we call \emph{Boolean CSP covering problems}.
Our focus is on 2-CSPs, where each constraint has at most 2 literals;
this includes the MaxCut, the Max2Sat, and the MaxDicut problems.
For each of these APX-complete problems, our framework transforms the
randomized approximation algorithms for the primal problem to
randomized approximation algorithms for their (also APX-complete)
duals while preserving the approximation factor.
In particular, it allows us to recover the same best known
approximation factors for the new problems.
For example, we provide a randomized (1/0.874)-approximation algorithm
for weighted fractional dicut covers.
Although the new problems have exponentially many variables, the
covers produced have small support and their approximation quality
relies on symmetric Grothendieck inequalities;
see~\cite{FriedlandLim2020}.
Our algorithms and analyses utilize Chernoff-type concentration
results and spectral sparsification.

We further describe how each APX-complete instance can be paired with
a dual APX-complete instance by solving a single semidefinite program,
unlike in usual scenarios where the dual is built syntactically from
the primal.
The SDP solutions yield, via a randomized sampling algorithm, primal
and dual feasible solutions along with a simultaneous certificate of
the approximation quality of both solutions.
Note that such a certificate has two primal-dual pairs involved: one
pair intractable, and the other pair tractable.
E.g.,
\begin{enumerate*}[label=(\roman*),]
\item MaxDicut and weighted fractional dicut cover,
\item the SDP relaxation of MaxDicut and its SDP dual.
\end{enumerate*}

Let \(D = (V, A)\) be a digraph.
For each \(U \subseteq V\), define \(\deltaout(U)\) as the set of arcs
leaving \(U\).
A \emph{dicut} is the set \(\deltaout(U)\) for some \(U \subseteq V\).
For arc weights \(w \in \Lp{A}\), the \emph{maximum dicut number of
  \((D, w)\)} is
\begin{equation*}
  \md(D,w) \coloneqq \max\setst{\iprodt{w}{\incidvector{\deltaout(U)}}}{U\in\Powerset{V}},
\end{equation*}
where \(\incidvector{\deltaout(U)} \in \set{0, 1}^A\) is the incidence
vector of \(\deltaout(U)\) and \(\Powerset{V}\) denotes the power set
of \(V\).
The vector of all-ones is \(\ones\).
The dual problem we consider is fractionally covering the arcs by
dicuts: for arc weights \(z \in \Lp{A}\), the \emph{fractional
  dicut-covering number of \((D, z)\)} is
\begin{equation*}
  \fdc(D, z)
  \coloneqq \min\setst[\bigg]{
    \iprodt{\ones}{y}
  }{
    y \in \Lp{\Powerset{V}},\,
    \sum_{U \subseteq V} y_U^{}\incidvector{\deltaout(U)} \ge z
  }.
\end{equation*}
\Cite{BrakensiekHuangPotechinZwick2023} obtained a randomized
\(\BHPZalpha\)-approximation for the maximum dicut problem, where
\(\BHPZalpha \approx 0.87446\).
Our framework yields the following result.


\begin{theorem}[Fractional Dicut Covering Theorem]
  \label{thm:maxdicut-cover}
  Fix \(\beta \in (0, \BHPZalpha)\).
  There is a randomized polynomial-time algorithm that, given a
  digraph \(D = (V,A)\) and \(z \in \Lp{A}\), computes
  \(w \in \Lp{A}\) and returns \(U \subseteq V\) and
  \(y \in \Lp{\Powerset{V}}\) with support size
  \(\card{\supp(y)} = O(\log \card{V})\) such that
  \( \sum_{S \subseteq V} y_S \incidvector{\deltaout(S)} \ge z\) holds
  with high probability (w.h.p.),
  \[
    \iprodt{\ones}{y} \le \tfrac{1}{\beta} \fdc(D,z),
    \qquad\text{and}\qquad
    \iprodt{w}{\incidvector{\deltaout(U)}} \ge \beta \md(D,w).
  \]
  Moreover, our algorithm returns a simultaneous certificate that each
  of \(U\) and \(y\) is within a factor of~\(\beta\) of the respective
  optimal value.
\end{theorem}

\begin{remark}[Primal-Dual Symmetry]
  Our results also allow one to start from an instance \((D,w)\) of the
  primal problem (i.e., MaxDicut) and the algorithm computes a dual
  instance \((D,z)\) of the fractional dicut-covering problem, along
  with \(\beta\)-approximate solutions for both and a simultaneous
  certificate.
  Analogous claims also apply to
  \cref{thm:max2sat-cover,thm:2-CSP-theorem-intro}.
\end{remark}

Let \((\constr,w)\) be an instance of the \emph{maximum
  2-satisfiability problem}, i.e., \(\constr\) is a set of disjunctive
2-clauses on two variables from \(x_1,\dotsc, x_n\), and \(w \in
\Lp{\constr}\) is a nonnegative weight vector.
Thus, each element of \(\constr\) has the form \(x_i \lor x_j\),
\(x_i \lor \overline{x_j}\), or
\(\overline{x_i} \lor \overline{x_j}\).
Let \(\assign \coloneqq \set{\false, \true}^n\) be the set of all
possible assignments for \((x_1,\dotsc, x_n)\).
For an assignment \(a\in \assign\), define
\(\val_{\constr}(a) \in \set{0,1}^{\constr}\) as the binary vector
indexed by \(\constr\) such that \((\val_{\constr}(a))_C = 1\) if
\(C\) is satisfied by \(a\), and \(0\) otherwise.
The goal is to find an assignment \(a\in \assign\) that maximizes the
inner product \(\iprod{w}{\val_{\constr}(a)}\).
Denote
\[
  \maxsat(\constr,w) \coloneqq
  \max\setst{\iprod{w}{\val_{\constr}(a)}}{a\in\assign}.
\]
The dual problem we consider is fractionally covering the clauses by
assignments: for weights \(z \in \Lp{\constr}\), the \emph{fractional
  2-sat covering number of \((\constr,z)\)} is
\begin{equation*}
  \fConstrC{2sat}(\constr,z)
  \coloneqq
  \min\setst[\bigg]{
    \iprodt{\ones}{y}
  }{
    y \in \Lp{\assign},\,
    \sum_{a \in \assign} y_a \val_{\constr}(a) \ge  z
  }.
\end{equation*}
\Cite{LewinLivnatZwick2002} provide a randomized
\(\LLZalpha\)-approximation algorithm for Max2Sat, where
\(\LLZalpha \approx 0.9401.\) Our framework yields the following
result.

\begin{theorem}[Fractional 2-Sat Covering Theorem]
  \label{thm:max2sat-cover}
  Fix \(\beta \in (0, \LLZalpha)\).
  There is a randomized polynomial-time algorithm that, given a set
  \(\constr\) of disjunctive 2-clauses on \(n\) variables and
  \(z \in \Lp{\constr}\), computes \(w \in \Lp{\constr}\) and returns
  an assignment \(a\in \assign\) and \(y \in \Lp{\assign}\) with
  \(\card{\supp(y)} = O(\log n)\) such that
  \(\sum_{a \in \assign} y_a \val_{\constr}(a) \ge z\) holds w.h.p.,
  \[
    \iprodt{\ones}{y} \le \tfrac{1}{\beta} \fConstrC{2sat}(\constr, z),
    \qquad\text{and}\qquad
    \iprod{w}{\val_{\constr}(a)} \ge \beta \maxsat(\constr,w).
  \]
  Moreover, our algorithm returns a simultaneous certificate that each
  of \(a\) and \(y\) is within a factor of~\(\beta\) of the respective
  optimal value.
\end{theorem}

Our results are general enough to include all forms of Boolean 2-CSPs.
A Boolean 2-CSP is a CSP where the variables \(x_1,\dotsc,x_n\) take
on Boolean values (i.e., \(\true\) or \(\false\)) and each constraint
involves only two variables.
Formally, we specify a Boolean constraint satisfaction problem using a
set \(\pred\) of binary predicate templates, i.e., functions from
\(\set{\false, \true}^2\) to \(\set{\false, \true}\).
We~assume throughout that the constant \(\false\) function is not in
\(\pred\).
Let \((\constr,w)\) be an instance of the \emph{(Boolean) maximum
  2-CSP problem}, i.e., each element of \(\constr\) is a function that
sends \(x \in\set{\false, \true}^n\) to \(f(x_i,x_j)\) for some
\(f \in \pred\) and \(i,j \in [n] \coloneqq \set{1,\dotsc,n}\), and
\(w \in \Lp{\constr}\).
We refer to an element of \(\constr\) as a \emph{\(\pred\)-constraint}
or just as a \emph{constraint}.
The \emph{maximum \(\pred\)-satisfiability number} of \((\constr,w)\)
is
\begin{equation*}
  \maxConstrSat{\pred}(\constr, w)
  \coloneqq \max\setst{\iprod{w}{\val_{\constr}(a)}}{a \in \assign}.
\end{equation*}
The dual problem is: for every \(z \in \Lp{\constr}\), the
\emph{fractional \(\pred\)-constraint covering number} of
\((\constr,z)\) is
\begin{align*}
  \fConstrC{\pred}(\constr, z)
  &\coloneqq
    \min\setst[\bigg]{
    \iprodt{\ones}{y}
  }{
    y \in \Lp{\assign},\,
    \sum_{a \in \assign} y_a \val_{\constr}(a) \ge  z
  }.
\end{align*}
By choosing distinct sets \(\pred\) one can formulate various
interesting problems.
By setting \(\pred \coloneqq \set{\overline{x_1} \lor x_2}\), we
recover the MaxDicut problem via \(\maxConstrSat{\pred}\) and the
fractional dicut-covering problem via \(\fConstrC{\pred}\).
Our Max2Sat results are recovered with
\(\pred \coloneqq \set{x_1 \lor x_2,\, \overline{x_1} \lor x_2,\,
  \overline{x_1} \lor \overline{x_2}}\).
Using these choices, \cref{thm:maxdicut-cover,thm:max2sat-cover} are
special cases a more general result from our framework, which we state
next.
The approximation factor \(\Xalpha\) that appears in the statement
will be defined shortly in \cref{eq:rounding-constant}; a
self-contained version of the result will be stated later as
\cref{thm:2-CSP-theorem}.

\begin{theorem}[Fractional \(\pred\)-Covering Theorem]
  \label{thm:2-CSP-theorem-intro}
  Let \(\pred\) be a set of predicates in two Boolean variables.
  Fix \(\beta \in (0, \Xalpha)\).
  There is a randomized polynomial-time algorithm that, given a set
  \(\constr\) of \(\pred\)-constraints on \(n\) variables and
  \(z \in \Lp{\constr}\), computes \(w \in \Lp{\constr}\) and returns
  an assignment \(a\in \assign\) and \(y \in \Lp{\assign}\) with
  \(\card{\supp(y)} = O(\log n)\) such that
  \(\sum_{a \in \assign} y_a \val_{\constr}(a) \ge z\) holds w.h.p.,
  \[
    \iprodt{\ones}{y} \le \tfrac{1}{\beta} \fConstrC{\pred}(\constr, z),
    \qquad\text{and}\qquad
    \iprod{w}{\val_{\constr}(a)} \ge \beta \maxConstrSat{\pred}(\constr, w).
  \]
  Moreover, our algorithm returns a simultaneous certificate that each
  of \(a\) and \(y\) is within a factor of~\(\beta\) of the respective
  optimal value.
\end{theorem}


%

Our framework builds on works by \nameandcite{GoemansWilliamson1995},
\nameandcite{Grothendieck1953}, and \nameandcite{Nesterov1998},
involving approximation results.
\Cite{GoemansWilliamson1995,Nesterov1998} both tackle the problem of
solving \(\max\setst{\qform{W}{s_U}}{U \in \Powerset{V}}\), where the
\(V\times V\) matrix \(W\) belongs to the positive semidefinite cone
\(\Psd{V}\) and \(s_U \coloneqq 2\ones_U-\ones \in \set{\pm1}^V\) is
the signed incidence vector of \(U \subseteq V\).
We introduce a parameterization for both the domain cone of matrices
\(W\) and the allowed subsets \(U\) of~\(V\).
Throughout \(V\) denotes a finite set and let
\(\Dist, \Cov \subseteq \Sym{V}\) be closed convex cones, where
\(\Sym{V}\) is the space of symmetric \(V\)-by-\(V\) matrices.
Let
\(\Fcal(\Dist) \coloneqq \setst{U\subseteq V}{\oprod{s_U}{s_U}
  \in\Dist}\) encode the feasible/allowed subsets of \(V\).
Our primal problem involves maximization of a quadratic form:
\begin{equation}
  \label{eq:maxq-A-K-def}
  \maxq_{\Dist,\Cov}(W)
  \coloneqq
  \max\setst{
    \qform{W}{s_U}
  }{
    U\in\Fcal(\Dist)
  },
  \mathrlap{
    \qquad
    \text{ for every }
    W \in \Cov.
  }
\end{equation}
Let \(\aff(\Cov)\) denote the smallest affine space containing \(\Cov\).
For each \(Z\) in the dual cone \(\Cov^* \coloneqq \setst{X \in
\aff(\Cov)}{\iprod{X}{Y} \geq 0 \text{ for each }Y \in \Cov}\) (where
we use the trace inner product), a vector \(y\in\Lp{\Fcal(\Dist)}\) is
a \emph{tensor sign cover for \(Z\)} if \(\sum_{U\in\Fcal(\Dist)}
y_U^{} \oprod{s_U}{s_U} \succeq_{\lift{\Cov^*}} Z\), where as usual
the notation \(A \succeq_{\lift{\Cov^*}} B\) means \(A - B \in
\lift{\Cov^*}\).
Here, we are denoting by \(\lift{\Cov^*} \coloneqq \setst{X \in
\Sym{n}}{\iprod{X}{Y} \ge 0 \text{ for each } Y \in \Cov}\) the dual
cone to \(\Cov\) in the potentially larger space of symmetric matrices
\(\Sym{n}\) --- see \cref{sec:conic-maxq-fevc}.
Our dual problem is to find a tensor sign cover \(y\) that minimizes
\(\iprodt{\ones}{y}\):
\begin{align}
  \label{eq:fevc-A-K-gaugef}
  \fevc_{\Dist,\Cov}(Z)
  &\coloneqq
  \min\setst[\bigg]{
      \iprodt{\ones}{y}
  }{
    y \in \Lp{\Fcal(\Dist)},\,
    \sum_{\mathclap{U \in \Fcal(\Dist)}} y_U^{} \oprod{s_U}{s_U}
    \succeq_{\lift{\Cov^*}} Z
  },
  \mathrlap{
    \quad
    \text{ for every }
    Z \in \Cov^*.
  }
\end{align}
The notation `\(\fevc\)' refers to fractional elliptope vertex cover.
Recall that the elliptope is the set
\(\Ecal^V \coloneqq \setst{Y \in \Psd{V}}{\diag(Y) = \ones}\), where
\(\diag \colon \Sym{V} \to \Reals^V\) extracts the diagonal, and its
vertices are \(\setst{\stensor{U}}{U \in \Powerset{V}}\);
see~\cite{LaurentPoljak1995}.
By fixing \(\Dist\) and varying \(\Cov\), it is clear that
\(\maxq_{\Dist,\Cov}\) always attributes the same value for an input
matrix, whereas \(\fevc_{\Dist,\Cov}\) defines a continuum of
relaxations, affecting feasibility via the constraint \(\sum_{U \in
\Fcal(\Dist)} y_U^{} \oprodsym{s_U}\succeq_{\lift{\Cov^*}} Z\) on the
tensor sign covers.
The smaller \(\Cov\) is, the weaker the constraint on the tensor sign
cover becomes.

\Cref{thm:maxdicut-cover,thm:max2sat-cover} describe SDP-based
approximation algorithms for fractional covering problems.
Covering problems, in general, proved to be difficult for tractable
SDP relaxations.
For some negative results on various SDP relaxations of vertex cover
problem, see for instance~\cite{GK1998,GMPT2010,BCGM2011}.
In those settings, the SDP relaxations considered fail to improve on
their much simpler LP-based counterparts, in terms of the
approximation ratio.
Thus, it is noteworthy that in our framework we obtain randomized
approximation algorithms that are tight under the UGC.
Another interesting feature of our results is that our conic covering
problems have an exponential number of variables (and computing their
optimal values is NP-hard) but we still are able to treat these
covering problems algorithmically, in polynomial time, and obtain
approximately optimal sparse covers.

Throughout the paper, we assume that
\(\Dist,\Cov \subseteq \Sym{n} \coloneqq \Sym{[n]}\) are closed
convex cones such that the following conditions hold:
\begin{equation}
  \label{eq:conic-assumptions}
  \begin{gathered}
    \Dist \subseteq \Psd{n},
    \qquad
    \Cov \subseteq \Dist^*,
    \qquad
    \int(\cone(\CUT^{\Dist})) \neq \emptyset,\\
    \set{0} \neq \Cov
    \text{ has a strictly feasible point},
  \end{gathered}
\end{equation}
where \(\CUT^{\Dist} \coloneqq \conv\setst{\stensor{U}}{U \subseteq
[n] ,\, \stensor{U} \subseteq \Dist}\), \(\conv\) is the convex hull,
\(\cone\) denotes the generated convex cone containing~\(0\), and
\(\int\) takes the interior.
We refer the reader to~\cref{sec:conic-maxq-fevc} for the definition
of strictly feasible point.
Set \(\Ecal(\Dist) \coloneqq \Ecal^{[n]} \cap \Dist\).
A \emph{randomized rounding algorithm \(\Round[]\) for \(\Dist\)} is
an indexed set \(\Round[] = (\Round[Y])_{Y \in \Ecal(\Dist)}\) of
matrix-valued random variables sampled from
\(\setst{\stensor{U}}{U \in \Fcal(\Dist)}\).
Define
\begin{equation}
  \label{eq:rounding-constant}
  \conealpha
  \coloneqq
  \inf_{Y \in \Ecal(\Dist)}
  \max\setst{\alpha \in \Lp{}}{
    \Ebb[\Round[Y]] \succeq_{\lift{\Cov^*}} \alpha Y
  },
\end{equation}
which we call the \emph{rounding constant for
  \((\Dist,\Cov,\Round[])\)}.
We shall drop the pair \(\conepair\) whenever they can be inferred by
context; in particular, the rounding constant \(\conealpha\) may
appear as~\(\Xalpha\).
We define a \emph{Grothendieck cover} for \(Z \in \Cov^*\) as a tensor
sign cover \(y\) for~\(Z\) such that
\(\iprodt{\ones}{y} \leq (1/\Xalpha) \fevc(Z)\).
Our~algorithms produce tensor sign covers \(y\) with approximation
factor \(\beta\) arbitrarily close to \(\Xalpha\); we also call such
vectors Grothendieck covers.

We show how to pair instances of the problems \(\maxq\) and \(\fevc\)
so that, given an instance \(W \in \Cov\) of \(\maxq\), we obtain an
instance \(Z \in \Cov^*\) of \(\fevc\) and we approximately solve both
instances simultaneously and provide a certificate for the
approximation factor of both solutions.
We do the same by starting with an instance of \(\fevc\).
Note from \cref{eq:fevc-A-K-gaugef} that feasible solutions can have
exponential support size.
The solutions produced by our algorithm have sparse support, with the
bound on the support size varying according to geometric properties of
the cone~\(\Cov\).
For the case that \(\Cov \subseteq \Psd{n}\), we~rely on spectral
sparsification results for positive semidefinite matrices
from~\cite{BatsonSS12a,deCarliSilvaHarveyEtAl2015}.

Our main results are the outcome of our framework powered by
primal-dual conic relaxations, randomized rounding algorithms together
with generalized Chernoff concentration results, and spectral
sparsifications methods.
We state our main results in
\cref{theorem:sdp-contained-algorithm,rem:polyhedral-algorithm}.
They output objects called \(\beta\)-certificates (see
\cref{def:beta-certificate}), where \(\beta\) is an approximation
factor, which are formed by feasible solutions for both problems,
together with a simultaneous certificate of their approximation
quality.
\begin{theorem}[Main Semidefinite Theorem]
  \label{theorem:sdp-contained-algorithm}
  Assume that \(\Cov \subseteq \Psd{n}\), and let \(\Round[]\) be a
  randomized rounding algorithm for \(\Dist\).
  Fix \(\beta \in (0, \Xalpha)\).
  There exists a randomized polynomial-time algorithm that, given an
  instance \(Z \in \Cov^*\) of \(\fevc\) as input, computes an
  instance \(W \in \Cov\) of \(\maxq\) and a \(\beta\)-certificate for
  \((W, Z)\) with high probability.
  Dually, there exists a randomized polynomial-time algorithm that,
  given an instance \(W \in \Cov\) of \(\maxq\) as input, computes an
  instance \(Z \in \Cov^*\) of \(\fevc\) and a \(\beta\)-certificate
  for \((W, Z)\) with high probability.
  Both algorithms take at most \(O(n^2 \log(n))\) samples from
  \(\Round[]\), and produce covers with \(O(n)\) support.
  If \(\Cov = \Psd{n}\), then \(O(n \log n)\) samples suffice.
\end{theorem}

\begin{remark}[Main Polyhedral Theorem]
  \label{rem:polyhedral-algorithm}
  We state in \Cref{thm:polyhedral-algorithm} our other main result,
  which is similar to \cref{theorem:sdp-contained-algorithm}, however,
  with a slightly different assumption on the cone~\(\Cov\) and it
  obtains better support size.
  The cone \(\Cov\) is the image \(\Acal(\Lp{d})\) for a linear map
  \(\Acal \colon \Reals^d \to \Sym{n}\), and the support size obtained
  is \(O(\log(n)+\log(d))\).
  \Cref{thm:2-CSP-theorem-intro} shall follow from this result.
\end{remark}

\subsection*{Additional Related Work}

In addition to the above cited references, here we mention some
additional related work.
In the continuum of the APX-complete duals, the one for MaxCut, called
fractional cut-covering problem, was previously studied: first, in the
special case that $z=\ones$, i.e.\ unweighted graphs,
see~\cite{Samal2006,NetoBen-Ameur2019}; then, in general (arbitrary
nonnegative weights $z$),
see~\cite{BenedettoProencadeCarliSilvaEtAl2023}.
We vastly generalize the results
of~\cite{BenedettoProencadeCarliSilvaEtAl2023} while keeping all the
desired properties.
Their results apply to the pair MaxCut and fractional cut covering,
which is a single pair of APX-complete problems in the wide swath of
APX-complete problem pairs covered here.

Part of the unification and generalization of the primal problems we
consider was proposed earlier~\cite{FriedlandLim2020}.
Their generalization is similar to the way we use the convex cone
\(\Cov\) and the generalized Grothendieck constant.
However, our framework is more general than that
of~\cite{FriedlandLim2020} in two ways: (i) we consider, as an
additional generalization, a set of convex cones \(\Dist\) restricting
the feasible region of the primal problem (this additional
generalization helps us achieve the best approximation ratios for the
duals of Max Boolean 2-CSPs); (ii) for every primal APX-complete
problem in our generalized domain we associate a dual conic covering
problem and provide randomized approximation algorithms which provide
approximate solutions to both problems.

Part of our development of the underlying theory leading to the
APX-complete duals is best explained via gauge
duality~\cite{Freund1987} and its interplay with conic duality.
A closely related concept is antiblocking duality
theory~\cite{Fulkerson1971,Fulkerson1972}.
The corresponding conic generalization of antiblocking duality
appeared previously in~\cite{Tind1974}.

\clearpage

\section{Framework for Generalized Cuts and Tensor Sign Covers and Certificates}

This section introduces our framework along with its theoretical
foundations.
Recall the assumptions~\cref{eq:conic-assumptions}.
We define relaxations for \(\maxq_{\Dist,\Cov}\) and \(\fevc_{\Dist,\Cov}\):
for every \(W \in \Cov\), set
\begin{subequations}
\label{eq:nu-def}
\begin{align}
  \label{eq:nu-A-K-supf-def}
  \nu_{\Dist,\Cov}(W)
  &\coloneqq \max\setst{
  \iprod{W}{Y}
  }{
    Y \in \Dist,\,
    \diag(Y) = \ones
  }\\*
  \label{eq:nu-A-K-gaugef-def}
  &\eqaligned \mathrlap{ \min\setst{
  \rho
  }{
    \rho \in \Lp{},\,
    x \in \Reals^n,\,
    \Diag(x) \succeq_{\Dist^*} W,\,
    \rho \geq \iprodt{\ones}{x}
  },}%
  \phantom{%
    \max\setst{
      \iprod{W}{Z}
    }{
      W \in \Cov,\,
      x \in \Reals^n,\,
      W \preceq_{\Dist^*} \Diag(x),\,
      \iprodt{\ones}{x} \le 1
    }.
  }
\end{align}
\end{subequations}
and, for every \(Z \in \Cov^*\),
\begin{subequations}
\label{eq:nu-polar-def}
\begin{align}
  \label{eq:nu-polar-A-K-gaugef}
  \nu^{\polar}_{\conepair}(Z)
  &\coloneqq \min\setst{
    \mu
  }{
    \mu \in \Lp{},\, Y \in \Dist,\,
    \diag(Y) = \mu\ones,\,
    Y \succeq_{\lift{\Cov^*}} Z
  }\\
  &\eqaligned \max\setst{
    \iprod{W}{Z}
  }{
    W \in \Cov,\,
    x \in \Reals^n,\,
    W \preceq_{\Dist^*} \Diag(x),\,
    \iprodt{\ones}{x} \le 1
  }.
\end{align}
\end{subequations}
Our algorithms rely on solving these relaxations and then sampling
using the feasible solutions found.
We show the following relation between \(\maxq\), \(\fevc\), \(\nu\),
and \(\nu^\polar\), and the rounding constant~\(\Xalpha\):

\begin{theorem}
\label{thm:maxq-fevc-nu-nupolar}
Let \(\Round[]\) be a randomized rounding algorithm for \(\Dist\).
We have that
\begin{alignat}{3}
    \label{eq:bound-primal}
    \Xalpha \cdot \nu(W)
    &\le \maxq(W)
    \le \nu(W)
    & \qquad &
    \text{for every } W \in \Cov;
    \\
    \label{eq:bound-dual}
    \nu^{\polar}(Z)
    &\le \,\,\,\fevc(Z)\,\,\,
    \le \tfrac{1}{\Xalpha} \cdot\nu^{\polar}(Z)
    & \qquad &
    \text{for every } Z \in \Cov^*.
  \end{alignat}
\end{theorem}
\begin{proof}
Note that \(\qform{W}{s_U} = \iprod{W}{\stensor{U}} \le \nu(W)\) for
every \(U \subseteq [n]\), so the second inequality in \cref{eq:bound-primal}
holds.
Let \(Y\) be a feasible solution of~\cref{eq:nu-A-K-supf-def}.
Since \(\Round[Y]\) has finite support, we have that \(\Ebb[\Round[Y]] = \sum_{U
\subseteq [n]} \prob\paren{\Round[Y] = \stensor{U}}\stensor{U}\), which implies
\begin{equation}
  \label{eq:expected-matrix-feasible}
  \Ebb[\Round[Y]]
  \in \conv(\setst{\stensor{U}}{U \in \Fcal(\Dist)})
  = \CUT^{\Dist}.
\end{equation}
Thus \cref{eq:bound-primal} follows from~\cref{eq:nu-A-K-supf-def}, as
\(\maxq(W) \ge \iprod{W}{\Ebb[\Round[Y]]} \ge \Xalpha \iprod{W}{Y}\).
Similarly, for every \(y \in \Lp{\Fcal(\Dist)}\) feasible in
\cref{eq:fevc-A-K-gaugef}, we have that \(\sum_{S \in \Fcal(\Dist)}
y_U^{}\stensor{U} \in \Dist\) is feasible in
\cref{eq:nu-polar-A-K-gaugef} with the same objective value, so the
first inequality in~\cref{eq:bound-dual} holds.
It is immediate from~\cref{eq:rounding-constant} that for every \(Y
\in \Ecal(\Dist)\), if \(\mu Y \succeq_{\lift{\Cov^*}} Z\), then \(\mu
\Ebb[\Round[Y]] \succeq_{\lift{\Cov^*}} \Xalpha Z\), so \(\fevc(Z) \le
\mu/\Xalpha\).
Hence \cref{eq:bound-dual} follows from \cref{eq:nu-polar-A-K-gaugef}.
\end{proof}

We remark that \cref{eq:bound-primal}  and \cref{eq:bound-dual} are
equivalent by gauge duality; see \cref{sec:conic-maxq-fevc} for more details.

Our discussion so far has focused exclusively on the matrix space.
Indeed, the definition~\cref{eq:rounding-constant} of~\(\Xalpha\), as well as the concentration results we will exploit are
naturally expressed in this context.
Yet, applications may require results on other spaces.
For example, an approximation algorithm for the fractional dicut
covering problem on a digraph \(D = (V, A)\) is about weights
in \(\Lp{A}\).
In our setting, this mapping between vectors and matrices is built into the cone \(\Cov\).
This is natural, as the cone \(\Cov\) is central to the covering constraint of~\cref{eq:fevc-A-K-gaugef}.
Let \(\Acal \colon \Reals^d \to \Sym{n}\) be a linear map.
We assume throughout the paper that \(\Acal(w) \coloneqq \sum_{i \in
[d]} w_i A_i\) for nonzero \(A_1,\dotsc,A_d \in \Sym{n}\).
Set \(\Cov \coloneqq \Acal(\Lp{d}) = \setst[\big]{\sum_{i \in [d]} w_i
  A_i}{ w \in \Lp{d}}\).
We have that, for every \(X, Y \in \Sym{n}\),
\begin{equation}
  \label{polyhedral-conic-ineq}
  X \preceq_{\lift{\Cov^*}} Y
  \text{ if and only if }
  \iprod{A_i}{X} \le \iprod{A_i}{Y}
  \text{ for every }
  i \in [d].
\end{equation}
One can succinctly encode the finite set of linear inequalities above
with the adjoint linear map \(\Acal^* \colon \Sym{n} \to
\Reals^d\), thus obtaining that \(X \preceq_{\lift{\Cov^*}} Y\) holds
if and only if \(\Acal^*(X) \le \Acal^*(Y)\).
This is similar to what is done in the entropy maximization setting; see,
e.g., \cite{SinghVishnoi2014}.
In particular, the linear map \(\Acal^*\) recovers relevant
marginal probabilities when working with random matrices.
With this setup, we move to \(\Lp{d}\) by defining,
for every \(w \in \Lp{d}\) and \(z \in \Lp{d}\),
\begin{alignat*}{3}
  \maxq_{\Dist, \Acal}(w)
  &\coloneqq \maxq_{\conepair}(\Acal(w)),
  &\qquad&&
  \fevc_{\Dist, \Acal}(z)
  &\coloneqq
  \min\setst{
    \fevc_{\conepair}(Z)
  }{
    Z \in \Cov^*,\,
    \Acal^*(Z) \ge z
  },\\
  \nu_{\Dist, \Acal}(w)
  &\coloneqq \nu_{\conepair}(\Acal(w)),
  &\qquad&&\nu_{\Dist, \Acal}^{\polar}(z)
  &\coloneqq
  \min\setst{
    \nu_{\conepair}^{\polar}(Z)
  }{
    Z \in \Cov^*,\,
    \Acal^*(Z) \ge z
  };
\end{alignat*}
see \cref{thm:projected-gauges} in \cref{sec:polyhedral-cones} for
details.
We highlight that
\begin{subequations}
  \label{eq:polyhedral-nu-polar}
\begin{align}
  \nu_{\Dist, \Acal}^{\polar}(z)
  &= \max\setst{
    \iprodt{z}{w}
  }{
    w \in \Lp{d},\,
    x \in \Reals^n,\,
    \Acal(w) \preceq_{\Dist^*} \Diag(x),\,
    \iprodt{\ones}{x} \le 1
  }\\
    \label{eq:polyhedral-nu-polar-min}
  &= \min\setst{
    \mu
  }{
    \mu \in \Lp{},\,
    Y \in \Dist,\,
    \diag(Y) = \mu\ones,\,
    \Acal^*(Y) \ge z
  }.
\end{align}
\end{subequations}
Hence, given \(z \in \Lp{d}\) as input, one can compute \(Z \in
\Cov^*\) such that \(\nu_{\Dist, \Acal}^{\polar}(z) = \nu_{\conepair}^{\polar}(Z)\), as
well as solving~\cref{eq:nu-polar-A-K-gaugef} for said \(Z\), by
solving a single convex optimization problem, namely~\cref{eq:polyhedral-nu-polar-min}.
In this way, we can both ``lift'' the vector \(w \in \Lp{d}\) to the matrix \(\Acal(w) \in
\Cov\), and \(z \in \Lp{d}\) to the matrix \(Z \in \Cov^*\), with no extra
algorithmic cost.

Next we discuss how to simultaneously certify the approximation
quality for instances \(W \in \Cov\) of \(\maxq\) and \(Z \in \Cov^*\)
of \(\fevc\).
A key observation is that
\begin{equation}
  \label{eq:cauchy-1}
  \maxq(W)\cdot \fevc(Z)
  \geq \iprod{W}{Z}
  \qquad
  \mathrlap{
  \text{
    for every \(W \in \Cov\) and \(Z \in \Cov^*\),
  }}
\end{equation}
which holds since
\(
  \iprod{W}{Z}
  \le \sum_{U \in \Fcal(\Dist)} y_U^{} \iprod{W}{\stensor{U}}
  \le \maxq(W) \iprodt{\ones}{y}
\)
for every feasible solution \(y \in \Lp{\Fcal(\Dist)}\) to
\cref{eq:fevc-A-K-gaugef}.
In the context of gauge duality, \cref{eq:cauchy-1} serves as a
\emph{weak duality result}.
Assumptions~\cref{eq:conic-assumptions} can be used to provide a
strong duality result: for every \(W \in \Cov\) there exists \(Z \in
\Cov^*\) such that equality holds in \cref{eq:cauchy-1}; and for every
\(Z \in \Cov^*\) there exists \(W \in \Cov\) such that equality holds
in \cref{eq:cauchy-1}.
Motivated by \cref{eq:cauchy-1}, we define \(\beta\)-pairings
and \(\beta\)-certificates.

\begin{definition}
\label{def:beta-pairing}
Let \(\beta \in \halfclosed{0,1}\).
A \emph{\(\beta\)-pairing on \((\Dist, \Cov)\)} is a pair \((W, Z) \in \Cov
\times \Cov^*\) such that there exist \(\rho, \mu \in \Lp{}\) with
\begin{equation}
  \label{eq:beta-pairing-def}
  \iprod{W}{Z}
  \,\overset{\text{\hyperref[eq:beta-pairing-def]{(\ref{eq:beta-pairing-def}a)}}}{\mathclap{=}}
  \rho\mu
  \qquad
  \text{and}
  \qquad
  \beta\rho\mu
  \overset{\text{\hyperref[eq:beta-pairing-def]{(\ref{eq:beta-pairing-def}b)}}}{\mathclap{\leq}}
  \maxq_{\conepair}(W)\mu
  \overset{\text{\hyperref[eq:beta-pairing-def]{(\ref{eq:beta-pairing-def}c)}}}{\mathclap{\leq}}
  \rho\mu
  \overset{\text{\hyperref[eq:beta-pairing-def]{(\ref{eq:beta-pairing-def}d)}}}{\mathclap{\leq}}
  \rho\fevc_{\conepair}(Z)
  \overset{\text{\hyperref[eq:beta-pairing-def]{(\ref{eq:beta-pairing-def}e)}}}{\mathclap{\leq}}\,
  \frac{1}{\beta}\rho\mu.
\end{equation}
If \(\Cov = \Acal(\Lp{d})\) for a linear map \(\Acal \colon
\Reals^d \to \Sym{n}\), we say that \((w, z) \in \Lp{d} \times
\Lp{d}\) is a \(\beta\)-pairing if \((\Acal(w), Z)\) is a
\(\beta\)-pairing for some \(Z \in \Cov^*\) such that \(\Acal^*(Z) \ge
z\).
We define an \emph{exact pairing on \((\Dist, \Cov)\)} to be a
\(1\)-pairing on \((\Dist, \Cov)\).
\end{definition}

Note that, for nonzero \(\rho\) and \(\mu\), this definition implies
that \(\beta\rho \leq \maxq(W) \leq \rho\) and \(\mu \leq \fevc(Z)
\leq (1/\beta) \mu\).
Thus, the definition of \(\beta\)-pairing establishes the idea of
simultaneous approximations.
We~need objects which algorithmically certify that a pair \((W, Z)\) is a
\(\beta\)-pairing.
For this, we use an analogue of \cref{eq:cauchy-1} for our relaxations:
\begin{equation}
  \label{eq:cauchy-2}
  \nu(W)\cdot \nu^{\polar}(Z)
  \geq \iprod{W}{Z}
  \qquad
  \mathrlap{
  \text{
    for every \(W \in \Cov\) and \(Z \in \Cov^*\).
  }}
\end{equation}
Similar to \cref{eq:cauchy-1}, inequality~\cref{eq:cauchy-2} follows from
\cref{eq:nu-A-K-supf-def} and \cref{eq:nu-polar-A-K-gaugef}.
The advantage of \cref{eq:cauchy-2} over \cref{eq:cauchy-1} is
that the quantities here are computable in polynomial time, and
by \cref{thm:maxq-fevc-nu-nupolar}, they are closely related to
the quantities in \cref{eq:cauchy-1}.

\begin{definition}
\label{def:beta-certificate}
    A \emph{\(\beta\)-certificate for \((W, Z) \in \Cov\times \Cov^*\)} is
a tuple \((\rho, \mu, U, y, x)\) such that
\begin{subequations}
  \renewcommand{\theequation}{\theparentequation.\roman{equation}}
  \begin{flalign}
    \label{item:cert-1}
    &\rho,\mu \in \Reals_+
      \text{ are such that }
      \rho\mu = \iprod{W}{Z},&&
    \\
    \label{item:cert-2}
    &U \in \Fcal(\Dist)
      \text{ is such that }
      \qform{W}{s_U} \ge \beta \rho,&&
    \\
    \label{item:cert-3}
    &y \in \Reals_+^{\Fcal(\Dist)}
      \text{ is such that }
      {\textstyle\sum_{U' \in \Fcal(\Dist)}} y_{U'}^{} \stensor{U'} \succeq_{\lift{\Cov^*}} Z
      \text{ and }
      \iprodt{\ones}{y} \le \tfrac{1}{\beta}\mu,
      \text{ and}&&
    \\
    \label{item:cert-4}
    &x \in \Reals^n
      \text{ is such that }
      \rho \geq \iprodt{\ones}{x}
      \text{ and }\Diag(x) \succeq_{\Dist^*} W.&&
  \end{flalign}
\end{subequations}
If \(\Cov = \Acal(\Lp{d})\) for a linear map \(\Acal \colon
\Reals^d \to \Sym{n}\), we say that \((\rho, \mu, U, y, x)\) is a
\(\beta\)-certificate for \((w, z) \in \Lp{d} \times \Lp{d}\) if
\((\rho, \mu, U, y, x)\) is   a \(\beta\)-certificate for
\((\Acal(w), Z)\) for some \(Z \in \Cov^*\) with \(\Acal^*(Z) \ge
z\).
\end{definition}

\begin{proposition}
    If there exists a \(\beta\)-certificate for \((W, Z)\), then \((W,Z)\) is a \(\beta\)-pairing.
\end{proposition}
\begin{proof}
  \hyperref[eq:beta-pairing-def]{%
    (\ref{eq:beta-pairing-def}a)},
  \hyperref[eq:beta-pairing-def]{%
    (\ref{eq:beta-pairing-def}b)},
  and \hyperref[eq:beta-pairing-def]{%
    (\ref{eq:beta-pairing-def}e)}
  follow immediately from \cref{item:cert-1}, \cref{item:cert-2}, and
  \cref{item:cert-3}, resp.
The connecting part \(\maxq(W)\mu \leq \rho\mu \leq \rho\fevc(Z)\) is
a combination of two notions of duality: conic duality (via \cref{eq:nu-def})
and conic gauge duality (via \cref{eq:cauchy-1}).
Item~\cref{item:cert-4} provides a feasible solution to
\cref{eq:nu-A-K-gaugef-def}, which implies \(\maxq(W) \le \nu(W) \le \rho\).
Hence
  \begin{equation*}
    \maxq(W)\mu
    \leq
    \rho\mu
    \overset{\text{\cref{item:cert-1}}}{\mathclap{=}}
    \iprod{W}{Z}
    \overset{\text{\cref{eq:cauchy-1}}}{\mathclap{\le}}
    \maxq(W)\cdot \fevc(Z)
    \leq
    \rho \fevc(Z).
    \qedhere
  \end{equation*}
\end{proof}

\section{Generalized Rounding Framework and Sparsification}
\label{sec:rounding}

Let \(\Round[]\) be a randomized rounding algorithm for \(\Dist\),
and let \(Y \in \Ecal(\Dist)\).
One can roughly see from \cref{eq:expected-matrix-feasible} that sampling from \(\Round[Y]\)
provides a feasible solution for \(\fevc\), as well as a feasible solution for \(\maxq\) in expectation.
However, such a solution may have
exponential support size.
Moreover, even in well-studied special cases like the Goemans and
Williamson algorithm, it is not known how to compute the marginal
probabilities exactly to obtain an expression for \(\Ebb[\Round[Y]]\).

We show how to obtain a Grothendieck cover by repeated sampling from a
randomized rounding algorithm so that we have polynomial support size
and the approximation ratio can be controlled with high probability.
We first treat the polyhedral case.

\begin{proposition}
\label{prop:polyhedral-cone-concentration}

Let \(\eps, \Chernoff \in (0, 1)\).
Let \(\Cov \coloneqq \Acal(\Lp{d})\) for a linear map
\(\Acal \colon \Reals^d \to \Sym{n}\).
Let \(X \colon \Omega \to \Cov\) be a random matrix such that
\(\Acal^*\paren{\Ebb[X]} \ge \eps \ones\).
Let \((X_t)_{t \in [T]}\) be i.i.d.\ random variables sampled
from~\(X\).
There is \(\psi_{\eps,\Chernoff} \in \Theta(1)\) such that, if \(T\geq
\psi_{\eps,\Chernoff}(\log(d) + \log(n))\), then \(\frac{1}{T}\sum_{t
\in [T]} X_t \succeq_{\lift{\Cov^*}} (1 - \Chernoff) \Ebb[X]\) with
probability at least \(1 - 1/n\).
\end{proposition}

The main argument in the proof of
\Cref{prop:polyhedral-cone-concentration}, which appears
in~\cref{sec:concentration}, relies on Chernoff's bound
for each generating ray of the cone, followed by
union bound on those rays.

\begin{proposition}
\label{prop:polyhedral-conic-sampling}

Let \(\eps, \Chernoff \in (0, 1)\).
Let \(\Cov \coloneqq \Acal(\Lp{d})\) for a linear map
\(\Acal \colon \Reals^d \to \Sym{n}\).
Let \(\Round[]\) be a randomized rounding algorithm for \(\Dist\).
Let \(Y \in \Ecal(\Dist)\) be such that \(\Acal^*(Y) \ge
\eps \ones\).
There exists a randomized polynomial-time algorithm
producing a Grothendieck cover \(y \in \Lp{\Fcal(\Dist)}\)
for \(Y\) {w.h.p.} such that the algorithm performs
at most \(T \coloneqq O(\log(d) + \log(n))\) samples from
\(\Round[Y]\), the support size \(\card{\supp(y)}\) is at most \(T\)
and  \(\iprodt{\ones}{y} \le ((1 - \Chernoff)\Xalpha)^{-1}\).
\end{proposition}

Beyond polyhedral cones, we present a rounding algorithm under
the assumption that \(\Cov \subseteq \Psd{n}\).
In this case, we leverage matrix Chernoff bounds to ensure
correctness of our algorithms with high probability.
We refer the reader to~\cref{sec:concentration} for a
complete proof.
The result is an application of \cite[Corollary~6.2.1]{Tropp2015},
which exploits results arising from a Matrix Chernoff bound with
respect to the positive semidefinite (Löwner) order.
We denote by \(\norm{X} \coloneqq
\max\set{\abs{\lambdamax(X)}, \abs{\lambdamin(X)}}\) the
\emph{spectral norm} on \(\Sym{n}\).

\begin{proposition}
\label{prop:Bernstein-whp-form}
Let \(X\) be a random matrix in \(\Sym{n}\) such that \(\norm{X} \le
\rho\) almost surely, and set \(\sigma^2 \coloneqq \norm{\Ebb[X^2]}\).
Let \((X_t)_{t \in [T]}\) be i.i.d.\ random variables sampled from \(X\).
There is \(\psi_{\Chernoff} = \Theta(1)\) such that, if \(T\geq
\psi_{\Chernoff}\max\set{\sigma^2,\rho}\log(n)\), then \(\Ebb[X] -
\Chernoff I \preceq \frac{1}{T} \sum_{t \in [T]} X_t \preceq \Ebb[X] +
\Chernoff I\) holds w.h.p..
\end{proposition}

The tensor sign covers we can obtain by directly applying
\cref{prop:Bernstein-whp-form} have polynomial support size.
To guarantee linear support size, we rely on the following
spectral sparsification result:

\begin{proposition}[{\cite[Corollary 10]{deCarliSilvaHarveyEtAl2015}}]
  \label{prop:sparsifying}
  Let \(Z \in \Psd{n}\).
  Let \(A_1, A_2 \dotsc, A_m \in \Psd{n}\) and \(c \in \Lp{m}\).
  Suppose that the semidefinite program
    \(
    \min\set[\big]
            {\iprodt{c}{y}: y \in \Lp{m},\,\sum_{i=1}^m y_iA_i \succeq Z}
  \)
  has a feasible solution \(y^*\).
  Let \(\bss \in (0,1)\).
  There is a deterministic polynomial-time algorithm that, given \(y^*\), and the matrices
  \(A_1,A_2\dotsc, A_m\) and \(Z\) as input, computes a feasible
  solution \(\bar{y}\) with at most \(O(n/\bss^2)\) nonzero entries and
  \(\iprodt{c}{\bar{y}} \leq (1+\bss) \iprodt{c}{y^{*}}\).
\end{proposition}

Thus, we obtain the following result,
which is proved in \cref{sec:concentration}.

\begin{proposition}
\label{prop:conic-sampling}

Let \(\Chernoff,\, \eps, \bss \in (0, 1)\).
Let \(\Cov \subseteq \Psd{n}\).
Let \(\Round[]\) be a randomized rounding algorithm for~\(\Dist\).
Let \(Y \in \Ecal(\Dist)\) be such that \(Y \succeq \eps
I\).
There exists a randomized polynomial time algorithm producing a Grothendieck  cover \(y \in
\Lp{\Fcal(\Dist)}\) for \(Y\) {w.h.p.}
such that the algorithm performs at most \(O(n^2 \log(n))\) samples from \(\Round[Y]\), the support size \(\card{\supp(y)} \) is \(O(n/\bss^2)\) and \(\iprodt{\ones}{y} \le (1 + \bss)((1 - \Chernoff)\Xalpha)^{-1}\).
\end{proposition}

\section{Simultaneous Approximation Algorithms}

The last ingredient of our algorithms is to ensure the feasible solutions
behave well with respect to our sampling results.
Both \cref{prop:polyhedral-conic-sampling,prop:conic-sampling} require
a numeric bound \(\eps\) on how interior to the cone the feasible
solutions are: either by requiring
\(Y \succeq \eps I\) or \(\Acal^*(Y) \ge \eps \ones\).
These assumptions are necessary: \cite{BenedettoProencadeCarliSilvaEtAl2023}
exhibits instances of the fractional cut-covering problem and optimal
solutions to the SDP relaxation that require, in expectation,
exponentially many samples to ensure feasibility.
For a fixed element of \(\Cov^*\) which is ``central'' enough,
we define perturbed versions of \cref{eq:nu-def,eq:nu-polar-def} whose
feasible regions exclude these ill-behaved matrices.
For concreteness, we assume that
\begin{equation}
  \label{eq:I-in-Dist}
  I
  \in \cone\setst{\stensor{U}}{U \in \Fcal(\Dist)}
  \subseteq \Dist,
\end{equation}
which can be easily verified in the examples we will work with.

We now describe one of the algorithms in \cref{theorem:sdp-contained-algorithm}.
Let \(\Round[]\) be a randomized rounding algorithm for \(\Dist\).
Let \(\beta \in (0, \Xalpha)\).
Assume we are given an instance \(Z \in \Cov^*\) of \(\fevc\) as input.
Then
\begin{enumerate}
  \item nearly solve the perturbed version of \cref{eq:nu-polar-def}
  to compute \((\mu,Y) \in \Reals_+ \times \Dist\) and \((W, x) \in \Cov \times \Reals^n\);
  \item sample \(O(n^2 \log n)\) times from \(\Round[Y]\) to obtain a Grothendieck cover \(y
  \in \Lp{\Fcal(\Dist)}\) for \(Z\);
  \item apply \cref{prop:sparsifying} to reduce the support size of
  \(y\) to \(O(n)\);
  \item choose \(U\) that maximizes \(\qform{W}{s_{U'}}\) among all
  \(U' \in \supp(y)\);
  \item output \(W\) and the \(\beta\)-certificate \((1, \mu, U, y, x)\).
\end{enumerate}
(Steps (1)--(3) involve errors terms that are chosen small
enough to guarantee our desired approximation factor \(\beta\).)
\Cref{prop:conic-sampling} proves the correctness of steps (2) and (3).
This is where we crucially exploit \(\Cov \subseteq \Psd{n}\), so that
concentration and sparsification results developed for positive
semidefinite matrices can be translated to the cone \(\Cov^*\).
That (4) will define a set \(U\) which is part of the \(\beta\)-certificate
follows from \(y\) being a good enough estimate: we have that \(\beta \rho \le
\qform{W}{s_U}\) since
\[
  \rho\mu
  = \iprod{W}{Z}
  \le \sum_{U' \in \Fcal(\Dist)} y_{U'}^{} \iprod{W}{\stensor{U'}}
  \le \iprod{W}{\stensor{U}} \iprodt{\ones}{y}
  \le \left(\iprodt{s_U}{W s_U}\right) \frac{1}{\beta} \mu.
\]
\Cref{sec:algorithmic-certificates} has the precise proofs.

The algorithm sketched above highlights an important part of our
framework.
For a given instance \(Z \in \Cov^*\) of \(\fevc\), we obtain from
the SDP solutions to \cref{eq:nu-polar-def} an instance \(W \in \Cov\)
of \(\maxq\), and we then certify \emph{the pair} \((W, Z)\).
This mapping among instances is something we now make explicit.
Define \(\GWOpt_{\conepair} \coloneqq \setst[\big]{
    (Z, W) \in \Cov \times \Cov^*
  }{
    \iprod{W}{Z} = \nu_{\conepair}^{}(W)\cdot \nu_{\conepair}^{\polar}(Z)
  }\).
One may prove that
\begin{align}
  \label{eq:N-def}
    \GWOpt=
    \setst[\Bigg]{
    (W, Z) \in \Cov \times \Cov^*
  }
  {
    \begin{array}{c}
      \exists (\mu,Y) \text{ feasible in~\cref{eq:nu-polar-A-K-gaugef} for }Z,\\
      \exists (\rho,x)\text{ feasible in~\cref{eq:nu-A-K-gaugef-def} for }W,\\
      \text{and } \iprod{W}{Z} = \rho \mu
     \end{array}
  }.
\end{align}
We invite the reader to compare the RHS of \cref{eq:N-def} with
the feasible regions of \cref{eq:nu-def} and \cref{eq:nu-polar-def}.
One may see solving either SDP as fixing one side of the pair of
instances and obtaining the other; i.e., as computing an element
of \(\optZ(W) \coloneqq \setst{Z \in \Cov^*}{(W, Z) \in
\GWOpt}\) when given \(W \in \Cov\) as input, or computing
an element of \(\optW(Z) \coloneqq \setst{W \in \Cov}{(W, Z) \in
\GWOpt}\) when given \(Z \in \Cov^*\) as input.
In~both cases, by solving a single (primal-dual pair of) SDP we obtain an element of
\(\GWOpt\) and the objects \((\rho, x)\) and \((\mu, Y)\)
which witness the membership.

We now address \cref{rem:polyhedral-algorithm},
in which the cone \(\Cov\) is polyhedral and not necessarily contained in~\(\Psd{n}\).
Here, we do not require the use of sparsification,
as the cover produced is already (very) sparse.

\begin{theorem}[Main Polyhedral Theorem]
\label{thm:polyhedral-algorithm}

Let \(\Cov \coloneqq \Acal(\Lp{d})\) for a linear map
\(\Acal \colon \Reals^d \to \Sym{n}\).
Assume~\cref{eq:I-in-Dist} and that \(\Acal^*(I) \ge \kappa \ones\)
for some positive \(\kappa \in \Reals\).
Let \(\Round[]\) be a randomized rounding algorithm for~\(\Dist\).
Fix \(\beta \in (0, \Xalpha)\).
There exists a randomized polynomial-time algorithm that, given
an instance \(z \in \Lp{d}\) of \(\fevc\) as input, computes an
instance \(w \in \Lp{d}\) of \(\maxq\) and a \(\beta\)-certificate
for \((w, z)\).
Dually, there exists a randomized polynomial-time algorithm that,
given an instance \(w \in \Lp{d}\) of \(\maxq\) as input, computes
an instance \(z \in \Lp{d}\) of \(\fevc\) and a \(\beta\)-certificate
for \((w, z)\).
Both algorithms output covers whose support size
is bounded by \(C \cdot (\log(d) + \log(n))\), where \(C \coloneqq
C(\kappa, \Xalpha, \beta)\) is independent of \(d\) and \(n\).
\end{theorem}

\section{Boolean 2-CSP}
\label{sec:boolean-csp}

Let \(U \subseteq \set{0} \cup [n]\) with \(0 \in U\).
Let \(x \colon [n] \to \set{\false, \true}\) be defined such that
\(x_i = \true\) if and only if \(i \in U \setminus \set{0}\).
Let \(i, j \in [n]\).
For any predicate \(P\), we let \(\Iverson{P} \in \set{0, 1}\) be
1 if the predicate \(P\) is \(\true\), and 0 otherwise.
\Cref{ssec:boolean-csp-appendix} defines matrices \(\Delta_{\pm i,
\pm j} \in \Sym{\set{0} \cup [n]}\) such that
\begin{equation*}
  \begin{aligned}
    \Iverson{\overline{x_i} \land \overline{x_j}}
    &= \tfrac{1}{4}\iprod{\Delta_{-i, -j}}{\stensor{U}},
    \qquad
    \Iverson{\overline{x_i} \land x_j}
    &= \tfrac{1}{4}\iprod{\Delta_{-i, +j}}{\stensor{U}},\\
    \Iverson{x_i \land \overline{x_j}}
    &= \tfrac{1}{4}\iprod{\Delta_{+i, -j}}{\stensor{U}},
    \qquad
    \Iverson{x_i \land x_j}
    &= \tfrac{1}{4}\iprod{\Delta_{+i, +j}}{\stensor{U}}.
  \end{aligned}
\end{equation*}
By decomposing a predicate as a disjunction of conjunctions, one can
write any Boolean function on two variables as a sum of these
matrices.  Thus, for any set \(\constr\) of constraints on two
variables, one can define a linear map
\(\Acal \colon \Reals^{\constr} \to \Sym{\set{0} \cup [n]}\) such that
\begin{equation}
  \label{eq:1}
  \iprod{\Acal(e_f)}{\stensor{U}} = \Iverson{f(x)} \text{ for every }
  f \in \constr;
\end{equation}
here, \(e_f \coloneqq \incidvector{\set{f}} \in \set{0,1}^{\constr}\)
is a canonical basis vector.  With this particular linear map
\(\Acal\), we say that \(\Cov\) is the \emph{polyhedral cone defined
  by \(\constr\)} if \(\Cov = \Acal(\Lp{\constr})\).  The definitions
are made so that
\begin{equation}
    \Ebb[\Round[Y]] \succeq_{\lift{\Cov^*}} \alpha Y
    \text{ if and only if }
    \prob(f(x) = \true) \ge \alpha \iprod{\Acal(e_f)}{Y}
    \text{ for every }
    f \in \constr,
\end{equation}
where \(x \in \assign\) is obtained from \(\Round[Y]\) in the
following way: let \(U  \subseteq \set{0} \cup [n]\) be such
that \(0 \in U\) and \(\stensor{U}\) was sampled from \(\Round[Y]\), and
define \(x \in \assign\) by \(x_i = \true\) if and only if \(i \in U \setminus \set{0}\).
If we set \(\Delta^n \coloneqq \conv\paren{ \bigcup_{i, j
\in [n]} \set{ \Delta_{\pm i, \pm j}}}\),
it is immediate that \(\Cov \subseteq \cone(\Delta^n)\).
The set
\(
  \Dist_{\Delta}
  \coloneqq
  \Psd{\set{0} \cup [n]} \cap \paren[\big]{\,
    \bigcup_{i, j \in [n]} \set{
      \Delta_{\pm i, \pm j}
    }
  }^*
\)
has been studied --- see e.g.,
\cite{Raghavendra2008,LewinLivnatZwick2002} ---, and these additional
inequalities are referred to as \emph{triangle inequalities}.
Since \(\CUT^{\Dist_{\Delta}} = \conv\setst{\stensor{U}}{U \subseteq
\set{0} \cup [n]}\) (see \cref{ssec:boolean-csp-appendix}) we have that
\(\Dist_{\Delta}^*  = \Psd{\set{0} \cup [n]} + \cone(\Delta^n)\).
This then ensures that \cref{eq:conic-assumptions} holds for \(\Dist_{\Delta}\)
and the polyhedral cone \(\Cov\) defined by \(\constr\).

\begin{theorem}[Fractional \(\pred\)-Covering Theorem]
\label{thm:2-CSP-theorem}
Let \(\pred\) be a set of predicates in two Boolean variables.
For every \(n \in \Naturals\), let \({\Round[]}_n\) be a randomized
rounding algorithm for \(\Dist_{\Delta} \subseteq \Sym{n + 1}\).
Let
\[
  \alpha
  \le \inf\setst[\bigg]{
    \alpha_{\Dist_{\Delta}, \Cov, {\Round[]}_n}
  }{
    \begin{array}{c}
      \constr \text{ set of \(\pred\)-constraints on \(n\) variables},\\
      \Cov \text{ polyhedral cone defined by } \constr.
    \end{array}
  },
\]
and fix \(\beta \in (0, \alpha)\).
There exists a randomized polynomial-time algorithm that,
given an instance \((\constr, z)\) of \(\fConstrC{\pred}\) as
input, computes \(w \in \Lp{\constr}\) and a \(\beta\)-certificate
for \((w, z)\).
Dually, there exists a polynomial-time randomized algorithm that,
given an instance \((\constr, w)\) of \(\maxConstrSat{\pred}\) as
input, computes \(z \in \Lp{\constr}\) and a \(\beta\)-certificate
for \((w, z)\).
Both algorithms take at most \(O(\log n)\) samples from \({\Round[]}_n\) and
produce covers with \(O(\log n)\) support size.
\end{theorem}

\begin{proof}[Proof of~\Cref{thm:2-CSP-theorem}]
Set \(\Dist \coloneqq \Dist_{\Delta}\) and let \(\Acal \colon
\Reals^{\constr} \to \Sym{\set{0} \cup [n]}\) be as in~\cref{eq:1}.
From~\cref{eq:triangle-ineq-preserves-sign-vectors} we have
that~\cref{eq:I-in-Dist} holds.
Since \(\Acal^*(\stensor{U}) = \val_{\constr}(x)\), using that
\(
  2^{-n} I = \sum\setst{\stensor{\set{0} \cup U}}{U \subseteq [n]}
\),
we see that \(\Acal^*(I)\) computes the marginal probability of
satisfying each constraint by  uniformly sampling an assignment
in \(\assign\).
As the constant \(\false\) function is not in \(\pred\), any
constraint is satisfied by at least \(1/4\) of the assignments.
Hence \(\Acal^*(I) \ge \tfrac{1}{4} \ones\).
Note that since \(\card{\pred} \le 16\), we have that
\(\log(\card{\constr}) = O(\log n)\).
\Cref{thm:polyhedral-algorithm} then ensures we can compute
\(\beta\)-certificates \((\rho, \mu, U, y, x)\) with
\(
  \card{\supp(y)}
  = O(\log(\card{\constr}) + \log(n))
  = O(\log(n))
\).
\end{proof}

\begin{proof}[Proof of \Cref{thm:maxdicut-cover}]
Set \(\pred \coloneqq \set{x_1 \land \overline{x_2}}\).
For every digraph \(D = (V, A)\), each arc \(uv\) can be mapped to a
constraint \(x_u \land \overline{x_v}\).
Hence, there exists \(\constr\) such that \(\md(D, w) =
\maxConstrSat{\pred}(\constr, w)\) and \(\fdc(\constr, z) =
\fConstrC{\pred}(\constr, z)\) for every \(w \in \Lp{A}\) and \(z \in \Lp{A}\).
\Cite{BrakensiekHuangPotechinZwick2023} --- see formulation after Proposition~2.4.
--- define \(\Round[]\) such that
\[
  \iprod{\Ebb[\Round[Y]]}{\tfrac{1}{4}\Delta_{+u, -v}}
  \ge \BHPZalpha \iprod{Y}{\tfrac{1}{4}\Delta_{+u, -v}}
  \text{ for every arc }
  uv \in A
  \text{ and }
  Y \in \Ecal(\Dist).
\]
Thus \(\Ebb[\Round[Y]] \succeq_{\lift{\Cov^*}} \BHPZalpha Y\), so
\Cref{thm:2-CSP-theorem} finishes the proof.
\end{proof}

\begin{proof}[Proof~of~\Cref{thm:max2sat-cover}]

Let \(\Round[Y] = \stensor{U}\) where \(U = \set{0} \cup \setst{i \in [n]}{x_i = \true}\)
for \(x \in \assign\) being sampled from the algorithm
defined by \nameandcite{LewinLivnatZwick2002}. Then
\begin{detailedproof}
  \iprod{\Ebb[\Round[Y]]}{\tfrac{1}{4}(\Delta_{-i, + j} + \Delta_{+i, -j} + \Delta_{+i, +j})}
  &= \prob(x_i \lor x_j)
  &&\text{{def.} of \(\Round\) and \(\Delta\)}\\*
  &\ge \LLZalpha
  \iprod{Y}{
    \tfrac{1}{4}\paren{
      3 + Y_{0i} + Y_{0j} + Y_{ij}
    }
  }
  &&\text{by~\cite{LewinLivnatZwick2002}}\\*
  &= \LLZalpha\iprod{Y}{\tfrac{1}{4}\paren{\Delta_{-i, + j} + \Delta_{+i, -j} + \Delta_{+i, +j}}}.
\end{detailedproof}
The mismatch between the expression in the second line
and the expression in \cite[Section~5]{LewinLivnatZwick2002}
arises from our modelling imposing \(x_0 = \true\), whereas \cite{LewinLivnatZwick2002}
impose \(x_0 = \false\).
The case for constraints \(\overline{x_i} \lor x_j\), \(x_i \lor \overline{x_j}\),
and \(\overline{x_i} \lor \overline{x_j}\) is analogous.
As this holds for every constraint, we have that
\(\Ebb[\Round[Y]] \succeq_{\lift{\Cov^*}} \LLZalpha Y\),
where \(\Cov\) is the polyhedral cone defined by \(\constr\).
Thus \(\LLZalpha \le \alpha_{\Dist, \Cov, \Round}\) and \cref{thm:2-CSP-theorem}
implies the statement.
\end{proof}

\section{Concluding Remarks and Future Directions}

Despite its generality, our framework still captures
several best possible and best known results.

A first aspect concerns the approximation constants
of the algorithms presented.
We refer to
\begin{equation}
  \label{eq:integrality-ratio}
  \intRatio_{\conepair}
  \coloneqq \min\setst[\bigg]{
    \frac{\maxq_{\conepair}(W)}{\nu_{\conepair}(W)}}{
    W \in \Cov
  }
  = \min\setst[\bigg]{
    \frac{\nu_{\conepair}^{\polar}(Z)}{\fevc_{\conepair}(Z)}
  }{
    Z \in \Cov^*
  }
\end{equation}
as the \emph{integrality ratio of \((\Dist, \Cov)\)}.
Equality between the two expressions above follows from gauge duality.
One may see that \(\intRatio_{\conepair}\) is the largest \(\beta\)
such that every \((W, Z) \in \GWOpt(\conepair)\) is a \(\beta\)-pairing.
Positivity of \(\intRatio_{\conepair}\) is a corollary of all norms on
a finite-dimensional vector space being equivalent.
It is more interesting then to consider families \(\Ical\) of triples
\((\Dist, \Cov, \Round[])\) where \(\Round[]\) is a randomized
rounding algorithm for \(\Dist\).
\Cref{thm:maxq-fevc-nu-nupolar} implies that
\begin{equation}
  \label{eq:integrality-ratio-rounding-constant}
  \intRatio_{\Ical}
  \coloneqq \inf\setst{\intRatio_{\conepair}}{(\Dist, \Cov, \Round[]) \in \Ical}
  \ge \inf\setst{\alpha_{\Dist, \Cov, \Round[]}}{(\Dist, \Cov, \Round[]) \in \Ical}
  \eqqcolon \alpha_{\Ical}.
\end{equation}
For example, if \(\Ical\) encodes all the cones arising from instances
of the maximum cut problem, to say that \(\intRatio_{\Ical} \ge \GWalpha\)
is to say we have a \(\GWalpha\)-approximation algorithm for the maximum
cut problem, and a \(1/\GWalpha\)-approximation algorithm for the fractional
cut covering problem.
Equality in \cref{eq:integrality-ratio-rounding-constant} indicates
that no better approximation algorithm can be obtained without
strengthening the formulation (by~changing \(\Dist\)) or restricting
the input instances (by~changing \(\Cov\)).
Whenever \(\Ical\) arises from instances related to a specific 2-CSP,
Raghavendra \cite[Corollary~1.5]{Raghavendra2008} shows that the
triangle inequalities (i.e., \(\Dist_{\Delta}\)) are enough, as there
exists a randomized rounding algorithm ensuring equality in
\cref{eq:integrality-ratio-rounding-constant}.
It is also known that \(2/\pi = \inf\setst{\intRatio_{\Psd{n}, \Psd{n}}}{
n \in \Naturals }\) \cite{FriedlandLim2020,AlonNaor2006,Grothendieck1953,Krivine1977}.
In this way, the algorithms in \Cref{theorem:sdp-contained-algorithm,%
thm:2-CSP-theorem} all have tight analyses.

One formulation of the ``equivalence between separation and optimization''
proved by \nameandcite{GrotschelLovaszEtAl1981} is that one
can compute a positive definite monotone gauge whenever one can compute
its dual.
In this way, whenever \(\nu_{\conepair}\) is the best polynomial-time
computable approximation to \(\maxq_{\conepair}\) under the Unique
Games Conjecture \cite{KhotKindlerEtAl2007} (and assuming P
\(\neq\) NP), the same immediately holds for \(\nu_{\conepair}^{\polar}\)
and \(\fevc_{\conepair}\).
In particular, \cite{Raghavendra2008} shows that, assuming
the UGC, it is NP-hard to obtain any approximation
algorithm for a Boolean 2-CSP with approximation factor better than
\(\intRatio_{\Dist_{\Delta}, \Cov}\).
Thus, the UGC implies that \Cref{thm:2-CSP-theorem}
is best possible unless P = NP.

The support size bounds in \Cref{thm:2-CSP-theorem,theorem:sdp-contained-algorithm}
are asymptotically tight.
If \(\Cov = \Psd{n}\), it is immediate that any
feasible~\(y\) in \cref{eq:nu-polar-A-K-gaugef} for \(Z = I\) has
\(\card{\supp(y)} \ge n\).
Hence the \(O(n)\) support size in \cref{theorem:sdp-contained-algorithm}
is best possible.
\Cite{BenedettoProencadeCarliSilvaEtAl2023} argues that
\(\card{\supp(y)} \ge \log(\chi(G))\) for every graph \(G = (V, E)\)
whenever \(\Cov = \Laplacian_G(\Lp{E})\), where \(\Laplacian_G(w)
\coloneqq \sum_{ij \in E} w_{ij}\soprod{(e_i - e_j)}\) is the
\emph{Laplacian of \(G\)}.
Hence the \(O(\log n)\) support size in \Cref{thm:2-CSP-theorem} is also
best possible.

Tightly related to the support size of the solutions we produce, is the
number of samples necessary to ensure a good enough cover with high
probability.
Although \Cref{theorem:sdp-contained-algorithm} shows that \(O(n^2 \log n)\)
samples suffice when \(\Cov \subseteq \Psd{n}\),
in specific cones we exploited \emph{conic concentration bounds} in
\cref{prop:polyhedral-cone-concentration,prop:psd-concentration} to
obtain better sampling bounds.
It is conceivable that other families of cones also admit better bounds.
E.g., \cite{SongZhang2022} offers conic concentration
for hyperbolicity cones.

Three of the natural generalizations of our framework not discussed here
are:
\begin{enumerate}[(i)]
\item extension to the complex field and Hermitian matrices,
\item extension of the intractable pairs defined by exponentially many
  constraints and exponentially many variables to a semi-infinite
  setting (infinitely many constraints in the intractable primal and
  infinitely many variables in the intractable gauge dual),
\item extension to handle general CSPs.
\end{enumerate}
The first two generalizations allow the treatment of many applications
in continuous mathematics and engineering, including some applications
in robust optimization and system and control theory.
The underlying theoretical results include as a special case the
Extended Matrix Cube Theorem \cite{B-TNR2003}.

\clearpage

\appendix

\section{maxq and fevc, and their Conic Relaxations}
\label{sec:conic-maxq-fevc}

We denote by \(\aff(S)\) the \emph{affine hull} of the set \(S
\subseteq \Sym{n}\), which is the intersection of all affine subspaces
of \(\Sym{n}\) containing \(S\).
Define the \emph{dual cone} of \(S \subseteq \Sym{n}\) as
\begin{equation*}
  S^*
  \coloneqq \setst{X \in \aff(S)}{
    \iprod{Y}{X} \ge 0
    \text{ for all }
    Y \in S
  }.
\end{equation*}
For full-dimensional convex sets, our definition matches the usual
definition of dual cone.
In general, the dual of a set is taken in its affine hull, analogous
to how the relative interior is taken with respect to the topology
induced in the affine hull of the set.
This becomes most relevant as~\cref{eq:conic-assumptions} allows for
cones \(\Cov\) which are not full dimensional.
Although for a convex cone \(\Cov \subseteq \Sym{n}\) the set
\(\aff(\Cov^*)\) may be strictly smaller than \(\aff(\Cov)\) --- and
hence \(\Cov^{**}\) may not be \(\Cov\) ---, if \(\Cov\) is a pointed
cone, then \(\Cov^{**} = \Cov\).

Let \(\Dist, \Cov \subseteq \Sym{n}\) be closed convex cones.
Recall that \(\CUT^{\Dist} \coloneqq \conv\setst{\stensor{U}}{U
  \subseteq V,\, \stensor{U} \in \Dist}\).
Assume
\begin{subequations}
  \label{eq:conic-assumptions-appendix}
  \begin{gather}
    \label{eq:dist-psd}
    \Dist \subseteq \Psd{n},\\
    \label{eq:non-negativity}
    \Cov \subseteq \Dist^*,\\
    \label{eq:Dist-slater-point}
    \int\paren{\CUT^{\Dist}} \neq \emptyset,\\
    \label{eq:Cov-slater-point}
    \Cov = \lift{\Cov} \cap \Null(\Lcal)
    \text{ and }
    \exists
    \slater{X} \in \int(\lift{\Cov}) \setminus \set{0}
    \text{ s.t. }
    \Lcal(\slater{X}) = 0,
  \end{gather}
\end{subequations}
where \(\lift{\Cov} \subseteq \Sym{n}\) is a closed convex cone and
\(\Lcal \colon \Sym{n} \to \Reals^k\) is a linear map for some \(k \in
\Naturals\).
Under these assumptions, we have that
\begin{equation}
  \label{eq:Cov-properties}
  \Cov \text{ is pointed and }
  \ri(\Cov) \setminus \set{0} \neq \emptyset.
\end{equation}
Note that \(\Dist^*\) is pointed, since \(\int(\Dist) \neq \emptyset\)
by~\cref{eq:Dist-slater-point}.
As \(\Dist^* \supseteq \Cov\) by~\cref{eq:non-negativity}, we conclude
\(\Cov\) is pointed.
The second part follows from~\cref{eq:Cov-slater-point}, as
\(\slater{X} \neq 0\) and \(\Null(\Lcal)\) is the smallest affine
subspace of \(\Sym{n}\) containing \(\Cov\).

Let \(\Dist, \Cov\) be closed convex cones such
that~\cref{eq:conic-assumptions-appendix} holds, and let \(\Lcal
\colon \Sym{n} \to \Reals^k\) and \(\lift{\Cov}\) be the linear
transformation and cone appearing in~\cref{eq:Cov-slater-point},
respectively.
We write
\begin{equation}
  \label{eq:lift-Cov-dual}
  \lift{\Cov^*}
  \coloneqq \lift{\Cov}^* + \Img(\Lcal^*)
  \subseteq \Sym{n}.
\end{equation}
If we denote by \(P \colon \Sym{n} \to \Sym{n}\) the orthogonal
projector onto \(\aff(\Cov) = \Null(\Lcal)\), then
\begin{equation}
  \label{eq:Cov-dual}
  \Cov^* = P(\lift{\Cov^*}).
\end{equation}
This relationship motivates the notation in~\cref{eq:lift-Cov-dual}:
it shows that \(\lift{\Cov^*}\) is a lifting of the cone \(\Cov^*\).
In our setting, we will have \(\aff(\Cov)\) as the \emph{instance
space}, where the inputs to our gauges arise from, and \(\Sym{n}\) as
the lifted space where optimization is performed.
In this way, both \(\Cov^*\) and its lifting \(\lift{\Cov^*}\) appear
throughout our developments.
From~\cref{eq:lift-Cov-dual} we have that \(\lift{\Cov^*} \supseteq
\Img(\Lcal^*) = \Null(\Lcal)^{\perp} = \aff(\Cov)^{\perp}\).
Hence
\begin{equation}
  \label{eq:aff-perp-subseteq-lift}
  \aff(\Cov)^{\perp} \subseteq \lift{\Cov^*}.
\end{equation}
From \cref{eq:non-negativity}, \cref{eq:Cov-slater-point},
and~\cref{eq:lift-Cov-dual} we have that
\begin{equation}
  \label{eq:lift-cov-dua-contains-dist}
  \Dist \subseteq  \lift{\Cov^*}.
\end{equation}
Since \(\Cov\) is pointed by~\cref{eq:Cov-properties}, we have that
\begin{equation}
  \label{eq:Cov-dual-involution}
  \Cov^{**} = \Cov.
\end{equation}
Finally, the orthogonal projector gives a convenient map from
\(\Dist\) to \(\Cov^*\), since
\begin{equation}
  \label{eq:Dist-to-dual-Cov}
  Y \succeq_{\lift{\Cov^*}} P(Y) \in \Cov^*
  \text{ for every }
  Y \in \Dist.
\end{equation}
Indeed, for every \(Y \in \Dist\) we have that \(P(Y) \in \Cov^*\), as
\(P(\Dist) \subseteq P(\lift{\Cov^*}) = \Cov^*\)
by~\cref{eq:Cov-dual} and~\cref{eq:lift-cov-dua-contains-dist}.
Moreover, \(Y - P(Y) \in \Null(\Lcal)^{\perp} = \Img(\Lcal^*) \subseteq
\lift{\Cov^*}\) by~\cref{eq:lift-Cov-dual},
so~\cref{eq:Dist-to-dual-Cov} holds.

Recall the definitions of \(\maxq_{\Dist,\Cov}\),
\(\fevc_{\conepair}\), \(\nu_{\conepair}\), and
\(\nu_{\conepair}^{\polar}\), along with conic dual formulations,
for~each \(W \in \Cov\) and \(Z \in \Cov^*\):
\begin{align}
  \label{eq:maxq-def-appendix}
  \maxq_{\Dist,\Cov}(W)
  &\!\coloneqq \mathrlap{%
    \max\setst{
      \iprodt{s_U}{W s_U}
    }{
      U\in\Fcal(\Dist)
    },%
  }%
  \phantom{%
    \max\setst[\big]{
      \iprod{Z}{X}
    }{
      X \in \Cov,\,
      \iprod{\stensor{U}}{X} \leq 1 \text{ for every }U \in \Fcal(\Dist)
    }
  }
\end{align}
\vspace*{-12pt}
\begin{subequations}%
  \label{eq:fevc-def-appendix}
  \begin{align}%
    \label{eq:fevc-gaugef-appendix}
    \phantom{\maxq_{\Dist,\Cov}(W)}
    \mathllap{\fevc_{\Dist,\Cov}(Z)}
    &\coloneqq \min\setst[\bigg]{
        \iprodt{\ones}{y}
      }{
      y \in \Lp{\Fcal(\Dist)},\,
      \sum_{\mathclap{U \in \Fcal(\Dist)}} y_U^{} \stensor{U}
      \succeq_{\lift{\Cov^*}} Z
      }
      \\
      \label{eq:fevc-supf}
      &\eqaligned \max\setst[\big]{
        \iprod{Z}{X}
      }{
        X \in \Cov,\,
        \iprod{\stensor{U}}{X} \leq 1 \text{ for every }U \in \Fcal(\Dist)
      };
  \end{align}
\end{subequations}
\vspace*{-12pt}
\begin{subequations}%
  \label{eq:nu-def-appendix}
  \begin{align}%
    \label{nu-supf-def}
    \phantom{\maxq_{\Dist,\Cov}(W)}
    \mathllap{\nu_{\Dist,\Cov}(W)}
    &\coloneqq \max\setst{
    \iprod{W}{Y}
    }{
      Y \in \Dist,\,
      \diag(Y) = \ones
    }
    \\
    \label{eq:nu-dual-def}
    &\eqaligned \mathrlap{%
      \min\setst{
        \iprodt{\ones}{x}
      }{
        x \in \Reals^n,\,
        \Diag(x) \succeq_{\Dist^*} W
      },
    }
  \phantom{%
    \max\setst[\big]{
      \iprod{Z}{X}
    }{
      X \in \Cov,\,
      \iprod{\stensor{U}}{X} \leq 1 \text{ for every }U \in \Fcal(\Dist)
    }
  }
  \end{align}
\end{subequations}
\vspace*{-12pt}
\begin{subequations}
  \label{eq:nu-polar-appendix-def}
  \begin{align}
    \label{eq:nu-polar-A-K-gaugef-appendix}
    \phantom{\maxq_{\Dist,\Cov}(W)}
    \mathllap{\nu^{\polar}_{\conepair}(Z)}
    &\coloneqq \min\setst{
      \mu
    }{
      \mu \in \Lp{},\, Y \in \Dist,\,
      \diag(Y) = \mu\ones,\,
      Y \succeq_{\lift{\Cov^*}} Z}
      \\
      \label{eq:nu-polar-A-K-supf}
    &\eqaligned \mathrlap{%
    \max\setst{
      \iprod{Z}{X}
    }{
      X \in \Cov,\,
      x \in \Reals^n,\,
      \Diag(x) \succeq_{\Dist^*} X,\,
      \iprodt{\ones}{x} \le 1
    }.
    }
  \phantom{%
    \max\setst[\big]{
      \iprod{Z}{X}
    }{
      X \in \Cov,\,
      \iprod{\stensor{U}}{X} \leq 1 \text{ for every }U \in \Fcal(\Dist)
    }
  }
  \end{align}
\end{subequations}
Our arguments rely on standard results on Conic Programming Duality
--- see, e.g., \cite[Chapter~7]{Nemirovski2024}.
In particular, a \emph{strictly feasible solution} to an optimization
problem is a feasible solution where every conic constraint is
satisfied by a point in the interior of the relevant cone.
By~\cref{eq:Dist-slater-point}, there exists
\begin{equation}
  \label{eq:slater-Y-Dist}
  \slater{y} \in \Lp{\Fcal(\Dist)}
  \text{ such that }
  \sum_{S \in \Fcal(\Dist)} \slater{y}_{U}^{}\stensor{U}
  \eqqcolon \slater{Y} \in \int(\Dist)
  \text{ and }
  \diag(\slater{Y}) = \ones.
\end{equation}
We may assume that \(\slater{y} > 0\).
Note that \(\alpha \slater{Y} - Z = \alpha(\slater{Y} -
\tfrac{1}{\alpha}Z)\in \int(\Dist) \subseteq \int(\lift{\Cov^*})\) for
large enough \(\alpha \in \Reals_{++}\).
Hence~\cref{eq:fevc-gaugef-appendix} has a strictly feasible solution.
From~\cref{eq:Cov-slater-point} one may
reformulate~\cref{eq:fevc-supf} into an equivalent problem with a
strictly feasible solution.
Conic Programming Strong Duality \cite[Theorem~7.2]{Nemirovski2024}
implies equality and attainment in~\cref{eq:fevc-def-appendix}.
Similarly, note that \(\slater{Y}\) is a strictly feasible solution
to~\cref{nu-supf-def}, whereas \(\slater{x} \coloneqq 2\lambdamax(W)
\ones\) is a strictly feasible solution for \cref{eq:nu-dual-def}, as
\(\Psd{n} \subseteq \Dist^*\) by~\cref{eq:dist-psd}.
Once again, Strong Duality ensures equality and attainment
in~\cref{eq:nu-def-appendix}.
A positive multiple of \((1, \slater{Y})\) is a strictly feasible
point to~\cref{eq:nu-polar-A-K-gaugef-appendix}.
Let \(\slater{X} \in \Cov\) be as in~\cref{eq:Cov-slater-point}.
Without loss of generality, assume that \(\lambdamax(\slater{X}) <
1\).
Then \((\tfrac{1}{2n} \ones, \tfrac{1}{2n}\slater{X})\) is a strictly
feasible solution to~\cref{eq:nu-polar-A-K-supf}.
Hence equality and attainment holds
in~\cref{eq:nu-polar-appendix-def}.

We will look at these functions through the lens of conic gauges, which are defined as follows:
\begin{definition}[Conic Gauges]
  Let \(\Ebb\) be an Euclidean space.
  Let \(\Cov \subseteq \Ebb\) be a closed convex cone.
  A function \(\varphi \colon \Cov \to \Reals_+\) is a \emph{gauge} if \(\varphi\) is positively homogeneous, sublinear, and \(\varphi(0) = 0\).
  The gauge \(\varphi\) is \emph{positive definite} if \(\varphi(x) > 0\) for each nonzero \(x \in \Cov\), and \(\varphi\)~is monotone if \(0 \preceq_{\Cov} x \preceq_{\Cov} y\) implies \(\varphi(x) \leq \varphi(y)\).
\end{definition}

\begin{definition}
  Let \(\varphi \colon \Cov \to \Lp{}\) be a positive definite
  monotone gauge.
  The \emph{dual} of \(\varphi\) is the positive definite monotone
  gauge \(\varphi^{\polar} \colon \Cov^* \to \Reals_+\) defined by
  \begin{equation}
      \varphi^{\polar}(y)
      \coloneqq
      \max\setst{\iprod{y}{x}}{x \in \Cov,\, \varphi(x) \leq 1}
\mathrlap{
      \qquad
      \text{for each }
      y \in \Cov^*.}
  \end{equation}
\end{definition}

Let \(\phi \colon \Cov \to \Lp{}\) be a positive definite monotone
gauge.
Whenever \(\Cov^{**} = \Cov\), --- in particular
whenever~\cref{eq:conic-assumptions-appendix} holds --- one can prove
that \(\phi^{\polar \polar} = \phi\).
We show that \(\maxq_{\conepair}\), \(\fevc_{\conepair}\),
\(\nu_{\conepair}\), and \(\nu^{\polar}_{\conepair}\)
are positive definite monotone gauges
and how they are related.

\begin{theorem}
\label{thm:conic-assumptions}

Let \(\Dist, \Cov \subseteq \Sym{n}\) be closed convex cones such
that~\cref{eq:conic-assumptions-appendix} holds.
Then
\begin{enumerate}[(i),itemsep=5pt]
    \item \(\maxq_{\conepair}\) and \(\fevc_{\conepair}\) are positive
      definite monotone gauges, dual to each other;
    \item\label{item:nu-gauges} \(\nu_{\conepair}\) and
      \(\nu_{\conepair}^{\polar}\) are positive definite monotone
      gauges, dual to each other;
    \item \(  \maxq_{\conepair} \le \nu_{\conepair}\) and
  \(\nu^{\polar}_{\conepair} \le \fevc_{\conepair}.\)
\end{enumerate}
\end{theorem}
\begin{proof}

The fact that \(\maxq_{\conepair}\) and \(\nu_{\conepair}\) are gauges
follows directly from their definitions in \cref{eq:maxq-def-appendix}
and \cref{nu-supf-def}.
The monotonicity of \(\maxq_{\conepair}\) and \(\nu_{\conepair}\) is a
direct consequence of~\cref{eq:non-negativity}.
Next we show that \(\maxq_{\conepair}\) is positive definite.
Let \(\slater{y} \in \Fcal(\Dist)\) and \(\slater{Y} \in \int(\Dist)\)
be as in~\cref{eq:slater-Y-Dist}.
Then
\[
  0 < \iprod{\slater{Y}}{W}
  = \sum_{U \in \Fcal(\Dist)} \slater{y}_U^{}\iprod{\stensor{U}}{W}
  \text{ for every nonzero } W \in \Cov \subseteq \Dist^*.
\]
Thus there exists \(U \in \Fcal(\Dist)\) such that
\(\maxq_{\conepair}(W) \ge \qform{W}{s_U} > 0\).
This implies that \(\maxq_{\conepair}\) is positive definite.
Since
\[\maxq_{\Dist, \Cov}(W) =
\max\setst{\qform{W}{s_U}}
{U \in \Fcal(\Dist)}
\le
\max\setst{\iprod{W}{Y}}
{
    Y \in \Dist,\,
    \diag(Y) = \ones
}
= \nu_{\Dist, \Cov}(W),\]
it follows that \(\nu_{\conepair}\) is positive definite. Thus,
\(\maxq_{\conepair}\) and \(\nu_{\conepair}\) are positive definite
monotone gauges such that \(\maxq_{\conepair} \le \nu_{\conepair}\).


The fact that \(\fevc_{\conepair}\) and \(\nu^{\polar}_{\conepair}\)
are gauges follows directly from \cref{eq:fevc-supf} and
\cref{eq:nu-polar-A-K-supf}.
We now prove \(\fevc_{\conepair}\) to be monotone.
Let \(P \colon \Sym{n} \to \Sym{n}\) denote the orthogonal projector
onto \(\aff(\Cov)\).
Let \(Z_0, Z_1 \in \Cov^*\) be such that \(Z_0 \preceq_{\Cov^*} Z_1\).
Let \(y \in \Lp{\Fcal(\Dist)}\) be such that \(\sum_{U \in
  \Fcal(\Dist)} y_U^{}\stensor{U} \succeq_{\lift{\Cov^*}} Z_1\).
By~\cref{eq:Cov-dual},
\begin{detailedproof}
  P\paren[\bigg]{
    \sum_{U \in \Fcal(\Dist)} y_U^{}\stensor{U}
  } \succeq_{\Cov^*} P(Z_1)
  = Z_1
  \succeq_{\Cov^*} Z_0
  = P(Z_0).
\end{detailedproof}
By~\cref{eq:Cov-dual}, there exists \(\hat{Y} \in \lift{\Cov^*}\) such
that
\(
  P\paren[\bigg]{
    \sum_{U \in \Fcal(\Dist)} y_U^{}\stensor{U}
    - Z_0
  } = P(\lift{Y}).
\)
Since \(\Null(P) = \aff(\Cov)^{\perp}\),
from~\cref{eq:aff-perp-subseteq-lift} we conclude
\[
  \sum_{U \in \Fcal(\Dist)} y_U^{}\stensor{U} - Z_0 - \lift{Y}
  \in \aff(\Cov)^{\perp}
  \subseteq \lift{\Cov^*},
\]
which ensures
\[
  \sum_{U \in \Fcal(\Dist)} y_U^{}\stensor{U}
  \succeq_{\lift{\Cov^*}} Z_0 + \lift{Y}
  \succeq_{\lift{\Cov^*}} Z_0.
\]
Hence, by~\cref{eq:fevc-gaugef-appendix}, we have that
\(\fevc_{\conepair}(Z_0) \le \fevc_{\conepair}(Z_1)\).
Similarly, let \(Y \in \Dist\) and \(\mu \in \Lp{}\) be such that
\(\diag(Y) = \mu\ones\) and \(Y \succeq_{\lift{\Cov^*}} Z_1\).
Then \cref{eq:Cov-dual} implies \(P(Y) \succeq_{\Cov^*} Z_1
\succeq_{\Cov^*} Z_0 = P(Z_0)\), so \(P(Y - Z_0) = P(\lift{X})\) for
some \(\lift{X} \in \lift{\Cov^*}\) by~\cref{eq:Cov-dual}.
Hence \(Y - Z_0 - \lift{X} \in \aff(\Cov)^{\perp} \subseteq
\lift{\Cov^*}\) by~\cref{eq:aff-perp-subseteq-lift}, and hence
\[
  Y \succeq_{\lift{\Cov^*}} Z_0 + \lift{X} \succeq_{\lift{\Cov^*}} Z_0.
\]
From~\cref{eq:nu-polar-A-K-gaugef-appendix} we conclude
\(\nu_{\conepair}^{\polar}\) is monotone.

By~\cref{eq:Cov-properties}, there exists \(\slater{X} \in \ri(\Cov)
\setminus \set{0}\).
Then \(\nu_{\conepair}(\slater{X}) > 0\), and hence we may assume
\(\nu_{\Dist, \Cov}(\slater{X}) = 1\).
We claim that
\begin{equation}
  \label{eq:nu-polar-pd}
  \iprod{Z}{\slater{X}} > 0
  \text{ for every nonzero }
  Z \in \Cov^*.
\end{equation}
Note that this implies via~\cref{eq:nu-polar-A-K-supf} that
\(\nu^{\polar}_{\Dist, \Cov}(Z) \ge \iprod{Z}{\slater{X}} > 0\) for
every nonzero \(Z \in \Cov^*\), so \(\nu^{\polar}_{\conepair}\) is
positive definite.
We now prove~\cref{eq:nu-polar-pd}.
Let \(\hat{Y} \in \lift{\Cov^*}\) and let \(P \colon \Sym{n} \to
\Sym{n}\) be the orthogonal projector onto \(\aff(\Cov)\).
Assume that \(P(\hat{Y}) \neq 0\).
By~\cref{eq:Cov-dual}, it suffices to prove
\(\iprod{P(\hat{Y})}{\slater{X}} > 0\).
Let \(\eps \in (0, 1)\).
Note that \(\slater{X} - \eps P(\hat{Y}) \in \aff(\Cov)\), since
\((1 - \eps)^{-1}\slater{X} \in \Cov\) and \(-P(\hat{Y}) \in
\aff(\Cov)\), so
\[
  \slater{X} - \eps P(\hat{Y})
  = (1 - \eps)\paren[\bigg]{\frac{1}{1 - \eps}\slater{X}}
  + \eps(-P(\hat{Y}))
  \in \aff(\Cov).
\]
Since \(\slater{X} \in \ri(\Cov)\), there exists \(\bar{\eps} > 0\) such
that \(\slater{X} - \bar{\eps} P(\hat{Y}) \in \Cov\).
Using that \(\iprod{\slater{X}}{\hat{Y}} =
\iprod{P(\slater{X})}{\hat{Y}} = \iprod{\slater{X}}{P(\hat{Y})}\)
and \(\iprod{P(\hat{Y})}{\hat{Y}} = \iprod{P(\hat{Y})}{P(\hat{Y})}\),
we use~\cref{eq:Cov-dual} to conclude
\begin{equation*}
  0
  \le \iprod{\slater{X} - \bar{\eps} P(\hat{Y})}{\hat{Y}}
  = \iprod{\slater{X}}{\hat{Y}} - \bar{\eps} \iprod{P(\hat{Y})}{\hat{Y}}
  = \iprod{\slater{X}}{P(\hat{Y})} - \bar{\eps}\iprod{P(\hat{Y})}{P(\hat{Y})},
\end{equation*}
so \cref{eq:nu-polar-pd} holds.
Since
\begin{equation*}
\begin{split}
  \fevc_{\Dist, \Cov}(Z)
  &= \min\setst[\bigg]{
    \iprodt{\ones}{y}
  }{
    y \in \Lp{\Fcal(\Dist)},\,
    \sum_{U \in \Fcal(\Dist)} y_U^{} \stensor{U}
    \succeq_{\lift{\Cov^*}} Z
  }\\
  &\ge
  \min\setst[\bigg]{
    \mu \in \Lp{}
  }{
    Y \in \Dist,\, \diag(Y) = \mu\ones,\,
    Y \succeq_{\lift{\Cov^*}} Z
  }= \nu^{\polar}_{\Dist, \Cov}(Z).
\end{split}
\end{equation*}
it follows that \(\fevc_{\conepair}\) is positive definite.
Thus, \(\fevc_{\conepair}\) and \(\nu^{\polar}_{\conepair}\) are
positive definite monotone gauges such that \(\fevc_{\conepair} \ge
\nu^{\polar}_{\conepair}\).

It is immediate from~\cref{eq:fevc-supf} that
\(\maxq_{\Dist, \Cov}^{\polar}  = \fevc_{\Dist, \Cov}\).
We have that
\begin{equation}
  \label{eq:nu-polar-polar}
  \nu_{\Dist, \Cov}(W)
  = \max\setst{
    \iprod{W}{Z}
  }{
    Z \in \Cov^*,\,
    \nu^{\polar}_{\Dist, \Cov}(Z) \le 1
  }
  \text{ for every }
  W \in \Cov.
\end{equation}
Moreover,
\begin{detailedproof}
  &\max\setst{
    \iprod{W}{Y}
  }{
    Y \in \Dist,\,
    \diag(Y) = \ones
  }\\
  &= \max\setst{
    \iprod{W}{P(Y)}
  }{
    Y \in \Dist,\,
    \diag(Y) = \ones
  }
  &&\text{since \(W \in \Cov \subseteq \Null(\Lcal)\)}\\
  &\le \max\setst{
    \iprod{W}{Z}
  }{
    Z \in \Cov^*,\,
    Y \in \Dist,\, \diag(Y) = \ones,\,
    Y \succeq_{\lift{\Cov^*}} Z
  }
  &&\text{by~\cref{eq:Dist-to-dual-Cov}}\\
  &\le \max\setst{
    \iprod{W}{Y}
  }{
    Y \in \Dist,\, \diag(Y) = \ones
  }.
  &&\text{since \(W \in \Cov \subseteq \lift{\Cov}\)}
\end{detailedproof}
Hence equality holds throughout, which
implies~\cref{eq:nu-polar-polar}
via~\cref{eq:nu-polar-A-K-gaugef-appendix}.
\end{proof}

\section{Polyhedral Cones}
\label{sec:polyhedral-cones}

Let \(\Acal \colon \Reals^d \to \Sym{n}\) be a linear map, and set
\(\Cov \coloneqq \Acal(\Lp{d})\).
We have that~\cref{eq:Cov-slater-point} always hold.
Indeed, we may write
\[
  \Acal(\Lp{d})
  = \setst{X \in \Sym{n}}{\Lcal(X) = 0,\, \Bcal(X) \ge 0}
\]
with \(\Lcal \colon \Sym{n} \to \Reals^k\) and \(\Bcal \colon \Sym{n}
\to \Reals^{\ell}\) linear transformations such that there exists
\(\slater{X} \in \Sym{n}\) such that \(\Lcal(X) = 0\) and \(\Bcal(X) >
0\).
Since \(\lift{\Cov} \coloneqq \setst{X \in \Sym{n}}{\Bcal(X) \ge 0}\)
has nonempty interior, we have that~\cref{eq:Cov-slater-point} holds.
Note further that \(\Null(\Lcal) = \Img(\Acal)\).
We further have that
\begin{equation}
  \label{eq:2}
  X \succeq_{\lift{\Cov^*}} Y
  \text{ if and only if }
  \Acal^*(X) \ge \Acal^*(Y)
  \quad
  \text{for every }
  X,\,Y \in \Sym{n}.
\end{equation}

\begin{proposition}
\label{prop:polyhedral-cone-pointed}
Let \(\Dist, \Cov\) be closed convex cones such
that~\cref{eq:conic-assumptions-appendix} holds,
where \(\Cov \coloneqq \Acal(\Lp{d})\) for a linear map
\(\Acal \colon \Reals^d \to \Sym{n}\) such that \(\Acal(e_i) \neq 0\)
for each \(i \in [d]\).
If \(w \in \Lp{d}\) is such that \(\Acal(w) = 0\), then \(w = 0\).
\end{proposition}
\begin{proof}

If there exists nonzero \(w \in \Lp{d}\) and such that \(\Acal(w) =
0\), then \(\Cov\) is not pointed,
contradicting~\cref{eq:Cov-properties}.
\end{proof}

\begin{theorem}
\label{thm:projected-gauges}
Let \(\Dist, \Cov \subseteq \Sym{n}\) be closed convex cones such
that~\cref{eq:conic-assumptions-appendix} holds,
where \(\Cov \coloneqq \Acal(\Lp{d})\) for a linear map
\(\Acal \colon \Reals^d \to \Sym{n}\) such that \(\Acal(e_i) \neq 0\)
for each \(i \in [d]\).
Then \(\nu_{\Dist,\Acal} \colon \Lp{d} \to \Reals\) defined by
\begin{equation}
  \label{eq:nu-projected-def}
  \nu_{\Dist, \Acal}(w)
  \coloneqq \nu_{\conepair}(\Acal(w))
  \mathrlap{
  \qquad
  \text{ for every }
  w \in \Lp{d}}
\end{equation}
is a positive definite monotone gauge, and its dual is the positive
definite monotone gauge
\begin{equation}
  \label{eq:nu-polar-projected-def}
  \begin{aligned}
    \nu_{\Dist,\Acal}^{\polar}(z)
    &=
    \min\setst{
    \nu_{\conepair}^{\polar}(Z)
    }{
    Z \in \Cov^*,\,
    \Acal^*(Z) \ge z
    }\\
    &= \min\setst{
      \mu
      }{
      \mu \in \Lp{},\,
      Y \in \Dist,\,
      \diag(Y) = \mu\ones,\,
      \Acal^*(Y) \ge z
      }
  \end{aligned}
\end{equation}
for every \(z \in \Lp{d}\).
Similarly, \(\maxq_{\Dist,\Acal} \colon \Lp{d} \to \Reals\) defined by
\begin{equation}
  \label{eq:maxq-projected-def}
  \maxq_{\Dist, \Acal}(w)
  \coloneqq \maxq_{\conepair}(\Acal(w))
  \mathrlap{
  \qquad
  \text{ for every }
  w \in \Lp{d}}
\end{equation}
is a positive definite monotone gauge, and its dual is the positive
definite monotone gauge
\begin{equation}
  \label{eq:fevc-projected-def}
  \begin{aligned}
    \fevc_{\Dist,\Acal}(z)
    &=
    \min\setst{
    \fevc_{\conepair}(Z)
    }{
      Z \in \Cov^*,\,
      \Acal^*(Z) \ge z
    }\\
    &= \min\setst[\bigg]{
      \iprodt{\ones}{y}
    }{
      y \in \Lp{\Fcal(\Dist)},\,
      \sum_{U \in \Fcal(\Dist)} y_U^{}\Acal^*(\stensor{U}) \ge z
    }
  \end{aligned}
\end{equation}
for every \(z \in \Lp{d}\).
\end{theorem}
\begin{proof}

We first prove~\cref{eq:nu-projected-def} to be a positive definite
monotone gauge.
As the composition of the gauge \(\nu_{\Dist, \Cov}\) with a
linear function, it is immediate that \(\nu_{\Dist,\Acal}\) is a gauge.
If \(0 \le w \le v\), then \(0 \preceq_{\Cov} \Acal(w) \preceq_{\Cov}
\Acal(v)\), so monotonicity of \(\nu_{\Dist, \Acal}\) follows from
the monotonicity part of \cref{thm:conic-assumptions}, \cref{item:nu-gauges}.
Let \(w \in \Lp{d}\) be such that \(\nu_{\Dist, \Acal}(w) = 0\).
Then \(\nu_{\conepair}(\Acal(w)) = 0\),
so \cref{thm:conic-assumptions} implies that \(\Acal(w) = 0\).
Hence \(w = 0\) by \cref{prop:polyhedral-cone-pointed}.
Thus \(\nu_{\Dist, \Acal}\) is a positive definite monotone
gauge.
Hence
\begin{detailedproof}
  \nu_{\Dist, \Acal}^{\polar}(z)
  &= \max\setst{
    \iprodt{z}{w}
  }{
    w \in \Lp{d},\, \nu_{\Dist, \Acal}(w) \le 1
  }\\
  &= \max\setst{
    \iprodt{z}{w}
  }{
    w \in \Lp{d},\, \nu_{\conepair}(\Acal(w)) \le 1
  }\\
  &= \max\setst{
    \iprodt{z}{w}
  }{
    w \in \Lp{d},\,
    x \in \Reals^n,\,
    \Diag(x) \succeq_{\Dist^*} \Acal(w),\,
    \iprodt{\ones}{x} \le 1
  }
  \\
  &= \min\setst{
    \mu
  }{
    \mu \in \Lp{},\,
    Y \in \Dist,\,
    \diag(Y) = \mu\ones,\,
    \Acal^*(Y) \ge z
  }.
\end{detailedproof}
Let \(\alpha > 0\) be such that \(\lambdamax(\Acal(\alpha \ones)) <
1\).
Then \((\slater{w}, \slater{x}) \coloneqq (\tfrac{\alpha}{2n}\ones,
\tfrac{1}{2n}\ones)\) is strictly feasible in the second to last
optimization problem, since \(\Diag(\slater{x}) - \Acal(\slater{w}) =
\tfrac{1}{2n} I - \tfrac{\alpha}{2n}\Acal(\ones) \in
\int(\Psd{n}) \subseteq \int(\Dist^*)\) and
\(\iprodt{\ones}{\slater{x}} < 1\).
For \(\slater{Y}\) as in~\cref{eq:slater-Y-Dist}, since \(\slater{Y}
\in \int(\Dist) \subseteq \int(\lift{\Cov^*})\), we have that
\(\Acal^*(Y) > 0\) from~\cref{eq:2}, and thus \(\Acal^*(\alpha
\slater{Y}) - z = \alpha(\Acal^*(\slater{Y}) - \tfrac{1}{\alpha} z) >
0\) for \(\alpha \in \Lp{}\) big enough.
Thus the last optimization problem is also strictly feasible.
Hence
\begin{detailedproof}
  \nu_{\Dist, \Acal}^{\polar}(z)
  &= \min\setst{
    \mu
  }{
    \mu \in \Lp{},\,
    Y \in \Dist,\,
    \diag(Y) = \mu\ones,\,
    \Acal^*(Y) \ge z
  }\\
  &= \min\setst{
    \mu
  }{
    \mu \in \Lp{},\,
    Y \in \Dist,\,
    Z \in \Cov^*,\,
    \diag(Y) = \mu\ones,\,
    Y \succeq_{\lift{\Cov^*}} Z,\,
    \Acal^*(Z) \ge z
  }\\
  &= \min\setst{
    \nu_{\conepair}^{\polar}(Z)
  }{
    Z \in \Cov^*,\,
    \Acal^*(Z) \ge z
  }.
\end{detailedproof}
The second equation holds because \((\mu, Y) \mapsto (\mu, Y, P(Y))\)
and \((\mu, Y, Z) \mapsto (\mu, Y)\) map feasible solutions between
both problems while preserving objective value by~\cref{eq:Dist-to-dual-Cov}.

That~\cref{eq:maxq-projected-def} is a positive definite monotone
gauge follows from \cref{thm:conic-assumptions} and
\cref{prop:polyhedral-cone-pointed} as above.
Hence
\begin{detailedproof}
  &\maxq_{\Dist, \Acal}^{\polar}(z)\\
  &= \max\setst{
    \iprodt{z}{w}
  }{
    w \in \Lp{d},\, \maxq_{\conepair}(\Acal(w)) \le 1
  }\\
  &= \max\setst{
    \iprodt{z}{w}
  }{
    w \in \Lp{d},\,
    \iprod{w}{\Acal^*(\stensor{U})} \le 1
    \text{ for every }
    U \in \Fcal(\Dist)
  }\\
  &= \min\setst[\bigg]{
    \iprodt{\ones}{y}
  }{
    y \in \Lp{\Fcal(\Dist)},\,
    \sum_{U \in \Fcal(\Dist)} y_{U}^{}\Acal^*(\stensor{U}) \ge z
  }
  &&\text{by LP Strong Duality}\\
  &= \min\setst[\bigg]{
    \iprodt{\ones}{y}
  }{
    Z \in \Cov^*,\,
    y \in \Lp{\Fcal(\Dist)},\,
    \sum_{U \in \Fcal(\Dist)} y_{U}^{}\stensor{U} \succeq_{\lift{\Cov^*}} Z,\,
    \Acal^*(Z) \ge z
  }
  &&\text{by~\cref{eq:Dist-to-dual-Cov} and~\cref{eq:2}}\\
  &= \min\setst{
    \fevc(Z)
  }{
    Z \in \Cov^*,\,
    \Acal^*(Z) \ge z
  }.
  &&\text{by~\cref{eq:fevc-gaugef-appendix}}
  &&\qedhere
\end{detailedproof}
\end{proof}

\section{Concentration Results}
\label{sec:concentration}
In this section we prove the concentration results in
\cref{sec:rounding}. First we prove the results concerning the
polyhedral case.

\begin{proof}[Proof of \cref{prop:polyhedral-cone-concentration}]

Set \(S \coloneqq \sum_{t \in [T]} X_t\).
Also, define \(x_t \coloneqq \Acal^*(X_t)\) for every \(t \in [T]\),
and set \(s \coloneqq \Acal^*(S)\).
We also denote \(x \coloneqq \Acal^*(X)\), where \(X\) is the
random matrix in the statement.
Then, by linearity of expectation,
\begin{equation*}
  \Ebb[s]
  = T \Ebb[x]
  = T \Ebb[\Acal^*(X)]
  = T \Acal^*(\Ebb[X])
  \ge T \eps \ones.
\end{equation*}
Let \(i \in [d]\).
Chernoff's bound and the previous inequality imply that
\begin{equation*}
  \prob\paren[\big]{s_i \le (1 - \Chernoff)\Ebb[s]_i}
  \le \exp\paren[\bigg]{
    - \Chernoff^2 \frac{\Ebb[s]_i}{2}
  }
  \le \exp\paren[\bigg]{
    - \frac{\Chernoff^2 \eps}{2}T
  }.
\end{equation*}
Hence, by the union bound,
\begin{detailedproof}
  \prob\paren{
    \exists i \in [d],\, s_i \le (1 - \Chernoff)\Ebb[s]_i
  }
  &\le d \exp\paren[\bigg]{
    - \frac{\Chernoff^2 \eps}{2} T
  }\\
  &\le \exp\paren[\bigg]{
    \log(d) - \frac{\Chernoff^2\eps}{2}
    \frac{2(\log(d) + \log(n))}{\Chernoff^2\eps}
  }
  = \frac{1}{n}.
\end{detailedproof}
Thus with probability at least \(1 - 1/n\) we have that \(s \ge (1 -
\Chernoff)\Ebb[s]\).
By~\cref{eq:2}, this event holds if and only if \(S
\succeq_{\lift{\Cov^*}} (1 - \Chernoff) \Ebb[S]\).
\end{proof}

\begin{proof}[Proof of \cref{prop:polyhedral-conic-sampling}]

Both \cref{polyhedral-conic-ineq} and~\cref{eq:rounding-constant}
imply that \(\Acal^*(\Ebb[\Round[Y]]) \ge \Xalpha \Acal^*(Y) \ge
\Xalpha \eps \ones\).
\Cref{prop:polyhedral-cone-concentration} implies that, with
probability at least \(1 - 1/n\),
\[
  \frac{1}{T} \sum_{t \in [T]} (\Round[Y])_t
  \succeq_{\lift{\Cov^*}} (1 - \Chernoff) \Ebb[\Round[Y]]
  \succeq_{\lift{\Cov^*}} (1 - \Chernoff) \Xalpha Y.
\]
Hence \(y \in \Lp{\Fcal(\Dist)}\) defined by \((1 - \Chernoff)\Xalpha \cdot y_U \coloneqq
\frac{1}{T}\card{\setst{t \in [T]}{(\Round[Y])_t = \stensor{U}}}\) for
every \(U \in \Fcal(\Dist)\) satisfies the desired properties.
\end{proof}

One case we treat separately is when \(\Cov = \Psd{n}\).
For this case, we use the following result by Tropp:
\begin{proposition}[{see \cite[Theorem~1.1]{Tropp2012}}]
\label{prop:Matrix-Chernoff}
Let \(\setst{X_t}{t \in T}\) be independent random matrices in
\(\Sym{n}\).
Let \(\rho \in \Reals\) be such that
\[
  0 \preceq X_t \preceq \rho I
  \text{ almost surely for every }
  t \in T.
\]
Set \(S \coloneqq \sum_{t \in T} X_t\).
Then for every \(\Chernoff \in (0, 1)\),
\[
  \prob\paren{
    \lambdamin(S)
    \le (1 - \Chernoff) \lambdamin(\Ebb S)
  }
  \le n \exp\paren[\bigg]{-\frac{\Chernoff^2}{2}\frac{\lambdamin(\Ebb S)}{\rho}}.
\]
\end{proposition}

\cref{prop:Matrix-Chernoff} weakens the upper bound from
\cite[Theorem~1.1]{Tropp2012} using that
\[
  \frac{\exp(-\Chernoff)}{(1 - \Chernoff)^{1 - \Chernoff}}
  \le \exp\paren[\bigg]{-\frac{\Chernoff^2}{2}},
\]
which follows from
\begin{equation*}
  \paren[\bigg]{1 - \frac{\Chernoff}{2}}\frac{\Chernoff}{1 - \Chernoff}
  = \Chernoff + \sum_{k = 2}^{\infty} \frac{\Chernoff^k}{2}
  \ge \Chernoff + \sum_{k = 2}^{\infty} \frac{\Chernoff^k}{k}
  = \log\paren[\bigg]{\frac{1}{1 - \Chernoff}}.
\end{equation*}

We prove the following result.
\begin{proposition}
\label{prop:psd-concentration}

Let \(\Chernoff \in (0, 1)\), let \(\tau,\,\rho \in \Lp{}\), and let
\(\bar{Y} \in \Sym{n}_{++}\).
Let \(X \colon \Omega \to \Sym{n}\) be a random matrix such that
\[
  0 \preceq X \preceq \rho \bar{Y}
  \text{ almost surely, and }
  \tau \bar{Y} \preceq \Ebb[X].
\]
Let \((X_t)_{t \in [T]}\) be independent identically distributed
random variables sampled from \(X\), for any
\begin{equation}
  \label{eq:psd-concentration-T}
  T \ge \ceil[\bigg]{
    \frac{4 \rho}{\Chernoff^2 \tau} \log(2n)
  }.
\end{equation}
Then, with probability at least \(1 - 1/2n\),
\[
  \frac{1}{T} \sum_{t \in [T]} X_t \succeq (1 - \Chernoff) \tau \bar{Y}.
\]
\end{proposition}
\begin{proof}

For every \(t \in [T]\), set \(Y_t \coloneqq
\bar{Y}^{-1/2}X_t \bar{Y}^{-1/2}\).
Then
\[
  0
  \preceq Y_t
  = \bar{Y}^{-1/2} X_t \bar{Y}^{-1/2}
  \preceq \rho \bar{Y}^{-1/2} \bar{Y} \bar{Y}^{-1/2}
  = \rho I
\]
for every \(t \in [T]\) almost surely.
Set \(Q \coloneqq \sum_{t \in [T]} Y_t\).
Since \(\Ebb[X] \succeq \tau \bar{Y}\),
\[
  \Ebb[Q]
  = \Ebb\sqbrac[\bigg]{\,\sum_{t \in [T]} Y_t}
  = \sum_{t \in [T]} \bar{Y}^{-1/2}\Ebb[X_t]\bar{Y}^{-1/2}
  \succeq T \tau I,
\]
which implies that \(\lambdamin(\Ebb[Q]) \ge T \tau\).
Hence
\begin{detailedproof}
  1 - \prob\paren[\bigg]{
    \frac{1}{T}\sum_{t \in [T]} X_t \succeq (1 - \Chernoff) \tau \bar{Y}
  }
  &=
  1 - \prob\paren{
    Q \succeq T(1 - \Chernoff)\tau I
  }\\
  &\le 1 - \prob\paren{Q \succeq (1 - \Chernoff)\lambdamin(\Ebb[Q]) I}
  &&\text{as \(\lambdamin(\Ebb[Q]) \ge T\tau\)}\\
  &\le \prob\paren{\lambdamin(Q) \le (1 - \Chernoff) \lambdamin(\Ebb[Q])}\\
  &\le 2n \exp\paren[\bigg]{-\Chernoff^2\frac{\lambdamin(\Ebb[Q])}{2\rho}}
  &&\text{by \cref{prop:Matrix-Chernoff}}\\
  &\le 2n\exp\paren[\bigg]{-\frac{\Chernoff^2 \tau}{2 \rho} T}
  &&\text{as \(\lambdamin(\Ebb[Q]) \ge T\tau\)}\\
  &\le \exp\paren[\bigg]{
    \log(2n) - \frac{\Chernoff^2\tau}{2\rho}\frac{4\rho\log(2n)}{\Chernoff^2\tau}
  }
  = \frac{1}{2n}.
  &&\text{by~\cref{eq:psd-concentration-T}}
  &&\qedhere
\end{detailedproof}
\end{proof}

For the more general case when \(\Cov\subseteq \Psd{n}\), we use the following result by Tropp which requires a bound on the spectral norm of the random matrix and uses its second moment:

\begin{theorem}[{see \cite[Corollary~6.2.1]{Tropp2015}}]
\label{thm:Bernstein}

Let \(T \in \Naturals\) be nonzero.
Let \(X\) be a random matrix in \(\Sym{n}\) such that
\(
  \norm{X} \le \rho
  \text{ almost surely}.
\)
Let \((X_t)_{t \in [T]}\) be i.i.d.\ random variables sampled from~\(X\).
Set \(\sigma^2 \coloneqq \norm{\Ebb[X^2]}\), set \(M
\coloneqq \Ebb[X]\), and set
\[
  E \coloneqq \frac{1}{T}\sum_{t = 1}^T X_t.
\]
Then for all \(\Chernoff \ge 0\),
\[
  \prob\paren{
    \norm{E - M}
    \ge \Chernoff
  }
  \le 2n\exp\paren[\bigg]{
    -T \frac{\Chernoff^2/2}{\sigma^2 + 2\rho\Chernoff/3}
  }.
\]
\end{theorem}

Using \cref{thm:Bernstein}, we are ready to prove our general
concentration result.

\begin{proof}[Proof of \cref{prop:Bernstein-whp-form}]
If \(\sigma^2 \ge (2/3)\rho\Chernoff\), then
\[
  -\frac{T}{2}\frac{\Chernoff^2}{\sigma^2 + (2/3)\Chernoff\rho}
  \le -\frac{T}{4}\frac{\Chernoff^2}{\sigma^2}
  \le -\frac{8\sigma^2\log(2n)}{\Chernoff^2}\frac{\Chernoff^2}{4 \sigma^2}
  = -2\log(2n).
\]
On the other hand, if \(\sigma^2 \le (2/3)\rho\Chernoff\), then
\[
  -\frac{T}{2}\frac{\Chernoff^2}{\sigma^2 + (2/3)\Chernoff\rho}
  \le -\frac{T}{4}\frac{\Chernoff^2}{(2/3)\Chernoff\rho}
  \le -\frac{16\rho\log(2n)}{3\Chernoff}\frac{3\Chernoff}{8\rho}
  = -2\log(2n).
\]
\Cref{thm:Bernstein} implies
\begin{equation*}
  \prob\paren[\bigg]{
    \bigg\| \dfrac{1}{T}\sum_{t \in [T]} X_t - \Ebb[X] \bigg\| \ge \Chernoff
  }
  \le
  2n\exp\paren[\bigg]{ -\frac{T}{2} \frac{\Chernoff^2}{\sigma^2 +
      (2/3)\rho\Chernoff}
  }
  \le 2n\exp\paren{
    -2\log(2n)
  }
  = \frac{1}{2n}.
  \qedhere
\end{equation*}
\end{proof}

Finally, we prove \Cref{prop:conic-sampling} by combining \cref{prop:Bernstein-whp-form} with
\cref{prop:sparsifying}.
\begin{proof}[Proof of~\Cref{prop:conic-sampling}]

Set \(Z \coloneqq Y^{-1/2} \Round[Y] Y^{-1/2}\).
For every \(U \subseteq [n]\),
\begin{equation*}
  \norm{Y^{-1/2} \stensor{U} Y^{-1/2}}
  = \iprod{\stensor{U}}{Y^{-1}}
  \quad
  \text{and}
  \quad
  \norm{(Y^{-1/2}\stensor{U}Y^{-1/2})^2}
  = \iprod{\stensor{U}}{Y^{-1}}^2.
\end{equation*}
Thus
\begin{equation*}
  \maxq(Y^{-1}) \ge \norm{Z} \text{  and }
  \maxq(Y^{-1})^2 I \succeq Z^2 \text{ almost surely}.
\end{equation*}
Set \(\sigma^2 \coloneqq (n/\eps)^2\) and \(\rho \coloneq n/\eps\).
As \(Y \succeq \eps I\), we have that \(Y^{-1} \preceq (1/\eps) I\),
so \(\maxq(Y^{-1}) \le n/\eps\).
Since \(n \ge 1 \ge (2/3)\Chernoff\Xalpha\eps\), we
have that
\[
  T
  \ge \ceil[\bigg]{
    \frac{8}{(\Chernoff\eps\Xalpha)^2}n^2 \log(2n)
  }
  \ge \frac{8}{(\Chernoff\Xalpha)^2}\frac{n^2}{\eps^2}\log(2n)
  = \frac{8\sigma^2\log(2n)}{(\Chernoff\Xalpha)^2}
  \ge \frac{16}{3}\frac{n}{\eps}\frac{1}{\Chernoff\Xalpha}\log(2n)
  = \frac{16\rho\log(2n)}{3\Chernoff\Xalpha}.
\]
Let \(\setst{Z_t}{t \in [T]}\) be {i.i.d.} random variables sampled from~\(Z\).
\Cref{prop:Bernstein-whp-form} implies that
\[
  Y^{-1/2}\Ebb[\Round[Y]]Y^{-1/2}
  - \gamma \Xalpha I
  \preceq Y^{-1/2}\paren[\bigg]{
    \frac{1}{T}\sum_{t \in [T]} Z_t
  }Y^{-1/2}
\]
with probability at least \(1 - 1/(2n)\).
Assume that this event holds.
Let \(y \in \Lp{\Fcal(\Dist)}\) be defined by
\(y_U \coloneqq \frac{1}{T}\card{\setst{t \in
    [T]}{Z_t = \stensor{U}}}\).
Then \(\iprodt{\ones}{y} = 1\) and
\[
  \sum_{U \in \Fcal(\Dist)} y_U^{} \stensor{U}
  = \dfrac{1}{T} \sum_{t \in [T]} Z_t
  \succeq \Ebb[\Round[Y]] - \Xalpha \Chernoff Y.
\]
\Cref{prop:sparsifying} implies we can compute in polynomial time
\(\tilde{y} \in \Lp{\Fcal(\Dist)}\) with support size
\(\card{\supp(\tilde{y})} \in O(n/\bss^2)\), such that
\(\iprodt{\ones}{\tilde{y}} \le 1 + \bss\), and \(\sum_{U \in
\Fcal(\Dist)} \tilde{y}_U^{} \stensor{U} \succeq \Ebb[\Round[Y]] -
\Xalpha\Chernoff Y\).
As \(\Cov \subseteq \Psd{n}\), we have that \(\Psd{n} \subseteq
\lift{\Cov^*}\), so by~\cref{eq:rounding-constant} we obtain
\begin{equation*}
  \sum_{U \in \Fcal(\Dist)} \tilde{y}_U \stensor{U}
  \succeq_{\lift{\Cov^*}}
  \Ebb[\Round[Y]] - \gamma \Xalpha Y
  \succeq_{\lift{\Cov^*}}
  \Xalpha Y - \gamma \Xalpha Y
  = (1 - \Chernoff)\Xalpha Y.
  \qedhere
\end{equation*}
\end{proof}

\section{Algorithmic Simultaneous Certificates}
\label{sec:algorithmic-certificates}

Let \(\Dist, \Cov \subseteq \Sym{n}\) be cones such
that~\cref{eq:conic-assumptions-appendix} holds.
Let \(\eps \in (0, 1)\).
Set, for every \(W \in \Cov\),
\begin{subequations}
  \begin{align}
    \label{eq:MAXQ-eps-def}
    \nu_{\eps,\conepair}(W)
    &\coloneqq (1 - \eps)\nu_{\conepair}(W) + \eps\iprod{I}{W}
    \\
    \label{eq:MAQQ-eps-def-dual}
    &\eqaligned
      \min\setst[\big]{
      \rho
      }{
      \rho \in \Reals_{+},\,
      x \in \Reals^n,\,
      \rho \geq (1-\eps)\iprodt{\ones}{x} + \eps\iprod{I}{W},\,
      \Diag(x) \succeq_{\Dist^*} W
    }.
  \end{align}
\end{subequations}
For every \(Z \in \Cov^*\), set
\begin{subequations}
  \label{eq:FEVC-eps-def}
  \begin{align}
    \label{eq:FEVC-eps-supf}
    \nu_{\eps,\conepair}^{\polar}(Z)
    &= \max\setst{
      \iprod{Z}{W}
      }{
      x \in \Reals^n,\,
      W \in \Cov,\,
      W \preceq_{\Dist^*} \Diag(x),\,
      (1 - \eps)\iprodt{\ones}{x} + \eps\iprod{I}{W} \le 1
      }
    \\
    \label{eq:FEVC-eps-gaugef}
    &= \min\setst{
      \mu
      }{
      \mu \in \Reals_+,\,
      Y \in \Sym{n},\,
      Y \succeq_{\Dist} \mu \eps I,\,
      Y \succeq_{\lift{\Cov^*}} Z,\,
      \diag(Y) = \mu \ones
      }.
  \end{align}
\end{subequations}
One may check that \(\nu_{\eps,\conepair}\) and
\(\nu_{\eps,\conepair}^{\polar}\) are positive definite monotone
gauges, dual to each other.
Let \(\sigma \in (0, 1)\) and set
  \begin{equation}
    \label{eq:N-approx-def}
    \GWOpt_{\eps, \sigma}(\Dist, \Cov)
    \coloneqq
    \setst*{
      (W, Z)
      \in \Cov \times \Cov^*
    }{
      \begin{array}{@{} l @{} l @{} l @{}}
        \exists (\mu,Y)
        & \text{ feasible for~\cref{eq:FEVC-eps-gaugef} }
        & \text{for }Z,
        \\
        \exists (\rho,x)
        & \text{ feasible for~\cref{eq:MAQQ-eps-def-dual} }
        & \text{for }W,
        \\
        \multicolumn{3}{c}{%
        \text{and } \iprod{W}{Z} \ge (1 - \sigma) \rho \mu
        }
      \end{array}
    }.
  \end{equation}
One may check that
\begin{equation}
  \label{eq:approx-Cauchy-Schwarz}
  \begin{gathered}
    \text{%
      if \((\rho, x)\) and \((\mu,Y)\) witness the membership \((W, Z)
      \in \GWOpt_{\eps, \sigma}(\conepair)\),
      then
    }\\
    (1-\sigma)\rho\mu
    \le
    \iprod{W}{Z}
    \le
    \rho\mu.
  \end{gathered}
\end{equation}

\begin{theorem}
  \label{thm:approx-bound-conversion}
  Let \(\eps \in \halfopen{0,1}\).
  Then,
  \begin{subequations}
    \begin{alignat}{3}
      \label{eq:MAXQ-eps-bound}
      (1 - \eps) \nu(W)
      &\le \nu_{\eps}(W)
      &&\le \nu(W),
      &
      \mathrlap{
        \qquad
        \text{for each \(W \in \Cov\)},
      }
      \\
      \label{eq:FEVC-eps-bound}
      \nu^{\polar}(Z)
      &\le \nu^{\polar}_\eps(Z)
      &&\le \frac{1}{1 - \eps} \nu^{\polar}(Z),
      &
      \mathrlap{
        \qquad
        \text{for each \(Z \in \Cov^*\)}.
      }
    \end{alignat}
  \end{subequations}
\end{theorem}
\begin{proof}

We have that \((1 - \eps) \nu(W) \le \nu_{\eps}(W)\) since
\(\nu_{\eps}(W) = (1 - \eps) \nu(W) + \eps \iprod{I}{W}\) and
\(\iprod{I}{W}\geq 0\) by~\cref{eq:I-in-Dist} since \(W \in \Cov
\subseteq \Dist^*\).

We have that \(\nu( W) \ge \iprod{I}{W}\) by~\cref{eq:I-in-Dist}.
Therefore, \(\nu_{\eps}(W) = (1-\eps) \nu(W) + \eps
\trace(W) \leq \nu(W)\).
\Cref{eq:FEVC-eps-bound} holds by duality.
\end{proof}

Let \(\Dist, \Cov \subseteq \Sym{n}\) be closed convex cones such
that~\cref{eq:conic-assumptions} holds, where \(\Cov\) is the polyhedral
cone defined by \(\Acal \colon \Reals^d \to \Sym{n}\).
Assume that~\cref{eq:I-in-Dist} holds.
Set
\begin{equation}
  \label{eq:nu-eps-projected-def}
  \nu_{\eps, \Dist, \Acal}(w)
  \coloneqq \nu_{\eps,\conepair}(\Acal(w))
  \text{ for every }
  w \in \Lp{d}.
\end{equation}
Then, for every \(z \in \Lp{d}\),
\begin{equation}
  \label{eq:nu-projected-psd}
  \begin{aligned}
    \nu_{\eps,\Cov,\Acal}^{\polar}(z)
    &= \min\setst{
      \mu \in \Lp{}
      }{
      Y \succeq_{\Dist} \eps\mu I,\,
      \diag(Y) = \mu\ones,\,
      \Acal^*(Y) \ge z
      }\\
    &= \max\setst{
      \iprodt{z}{w}
      }{
      w \in \Lp{d},\,
      x \in \Reals^n,\,
      (1 - \eps)\iprodt{\ones}{x} + \eps\trace(\Acal(w)) \le 1,\,
      \Diag(x) \succeq_{\Dist^*} \Acal(w)
      }.
  \end{aligned}
\end{equation}
Let \((\slater{\mu}, \slater{Y})\) be as in~\cref{eq:slater-Y-Dist}.
Since \(I \in \Dist \subseteq \Cov^*\) by~\cref{eq:I-in-Dist}
and~\cref{eq:conic-assumptions-appendix}, it follows
from~\cref{polyhedral-conic-ineq} that \(\Acal^*(I) \ge 0\).
Hence
\[
  \Acal^*((1 - \eps)\slater{Y} + \eps \slater{\mu}I)
  \ge (1 - \eps) \Acal^*(\slater{Y})
  > 0.
\]
We thus conclude that \((1 - \eps)\slater{Y} + \eps\slater{\mu} I\) is
strictly feasible in the first optimization problem
in~\cref{eq:nu-projected-psd}.
The second problem in~\cref{eq:nu-projected-psd} is also strictly
feasible, as one can see by setting \((\slater{w}, \slater{x})
\coloneqq (\tfrac{\alpha}{3n} \ones, \frac{1}{3n} \ones)\) for
\(\alpha \in \Reals_{++}\) such that \(\lambdamax(\Acal(\alpha \ones))
< 1\).
Hence the second equality in~\cref{eq:nu-projected-psd} and attainment
of both problems follow from Strong Duality.
For the first equality in~\cref{eq:nu-projected-psd}, note that
\begin{detailedproof}
  \label{eq:nu=polar-projected-def}
  &\nu_{\eps}^{\polar}(z)\\
  &= \min\setst{
    \nu_{\eps}^{\polar}(Z)
  }{
    Z \in \Cov^*,\,
    \Acal^*(Z) \ge z
  }\\
  &= \min\setst{
    \mu \in \Lp{}
  }{
    Z \in \Cov^*,\,
    \Acal^*(Z) \ge z,\,
    Y \in \Sym{n},\,
    Y \succeq_{\Dist} \eps\mu I,\,
    Y \succeq_{\lift{\Cov^*}} Z,\,
    \diag(Y) = \mu\ones
  }
  &&\text{by~\cref{eq:FEVC-eps-gaugef}}\\
  &\ge \min\setst{
    \mu \in \Lp{}
  }{
    Y \in \Sym{n},\,
    Y \succeq_{\Dist} \eps\mu I,\,
    \diag(Y) = \mu\ones,\,
    \Acal^*(Y) \ge z
  },
\end{detailedproof}
as \(Y \succeq_{\lift{\Cov^*}} Z\) implies \(\Acal^*(Y) \ge
\Acal^*(Z)\) by~\cref{eq:2}.
Equality follows from~\cref{eq:Dist-to-dual-Cov} and \(\Acal^*(Y) =
\Acal^*(P(Y))\) for every \(Y \in \Sym{n}\), where \(P \colon \Sym{n}
\to \Sym{n}\) is the orthogonal projector on \(\aff(\Cov)\).

\begin{proposition}
\label{prop:polyhedral-perturbed-versions}

Let \(\Dist, \Cov \subseteq \Sym{n}\) be closed convex cones such
that~\cref{eq:conic-assumptions} holds, where \(\Cov \coloneqq
\Acal(\Lp{d})\) for a linear map \(\Acal \colon \Reals^d
\to \Sym{n}\).
Assume that~\cref{eq:I-in-Dist} holds.
Let \(z \in \Lp{d}\).
If \((\mu, Y)\) and \((w, x)\) are feasible solutions
to~\cref{eq:nu-projected-psd} such that
\[
  (1 - \sigma)\mu
  \le \iprodt{z}{w}
  \le \mu,
\]
then
\[
  (1, x) \text{ and }  (\mu, Y)
  \text{ witness the membership }
  (\Acal(w),Y) \in \GWOpt_{\eps, \sigma}(\Dist, \Cov).
\]
\end{proposition}
\begin{proof}

It is immediate that \((1, x)\) is feasible
in~\cref{eq:MAQQ-eps-def-dual} for \(W \coloneqq \Acal(w)\).
It is also clear that \((\mu, Y)\) is feasible
in~\cref{eq:FEVC-eps-gaugef} for \(Z \coloneqq Y\).
The proof follows from
\[
  (1 - \sigma)\mu
  \le \iprodt{w}{z}
  \le \iprod{w}{\Acal^*(Y)}
  = \iprod{\Acal(w)}{Y}.
  \qedhere
\]
\end{proof}

\begin{proposition}
\label{prop:certificate-auxiliary}

Let \(\eps, \sigma, \bss \in \halfopen{0, 1}\).
Let \(\Dist, \Cov \subseteq \Sym{n}\) be such
that~\cref{eq:conic-assumptions} holds.
Let \(W, Z \in \Sym{n}\) be nonzero.
Let \((\bar{\rho}, x)\) and \((\bar{\mu}, Y)\) witness the membership
\((W, Z) \in \GWOpt_{\eps, \sigma}(\Dist, \Cov)\).
Set \(\rho \coloneqq (1 - \eps)^{-1}\bar{\rho}\) and \(\mu \coloneqq
\rho^{-1}\iprod{Z}{W}\).
Let \(p \in \Lp{\Fcal(\Dist)}\) be such that \(\iprodt{\ones}{p} \le 1
+ \bss\) and
\begin{equation*}
  \sum_{U \in \Fcal(\Dist)} p_U^{}\stensor{U} \succeq_{\lift{\Cov^*}} \frac{1}{\bar{\mu}}Y.
\end{equation*}
Set \(\beta \coloneqq (1 - \eps)(1 - \sigma)/(1 + \bss)\).
Then
\[
  \qform{W}{s_V} \ge \beta \rho
  \text{ and }
  Z \preceq_{\lift{\Cov^*}} \frac{1}{\beta} \sum_{U \in \Fcal(\Dist)} p_U^{}\stensor{U}
\]
for \(V \coloneqq \argmax\setst{\qform{W}s_U}{U \in \supp(p)}\).
\end{proposition}
\begin{proof}
Set \(V \coloneqq \argmax\setst{\qform{W}{s_U}}{U \in \supp(p)}\).
Then
\begin{detailedproof}
  \iprod{\stensor{V}}{W}
  &\ge \sum_{U \in \Fcal(\Dist)} \frac{p_U}{\iprodt{\ones}{p}}\iprod{\stensor{U}}{W}\\
  &\ge \frac{1}{\bar{\mu}\iprodt{\ones}{p}}\iprod{Y}{W}\\
  &\ge \frac{1}{\bar{\mu}\iprodt{\ones}{p}}\iprod{Z}{W}
  &&\text{by~\cref{eq:FEVC-eps-gaugef}}\\
  &\ge \frac{1 - \sigma}{\iprodt{\ones}{p}}\bar{\rho}
  &&\text{by~\cref{eq:N-approx-def}}\\
  &= \frac{(1 - \sigma)(1 - \eps)}{\iprodt{\ones}{p}}\rho\\
  &\ge \frac{(1 - \sigma)(1 - \eps)}{1 + \bss}\rho.
\end{detailedproof}
Moreover,
\[
  \mu
  = \rho^{-1}\iprod{Z}{W}
  \ge \rho^{-1}(1 - \sigma)\bar{\rho}\bar{\mu}
  = (1 - \sigma)(1 - \eps)\bar{\mu}.
\]
Hence
\begin{detailedproof}
  Z
  &\preceq_{\lift{\Cov^*}} Y
  &&\text{by~\cref{eq:FEVC-eps-gaugef}}\\
  &\preceq_{\lift{\Cov^*}} \bar{\mu} \sum_{U \in \Fcal(\Dist)} p_U^{}\stensor{U}\\
  &\preceq_{\lift{\Cov^*}} \frac{\mu}{(1 - \sigma)(1 - \eps)} \sum_{U \in \Fcal(\Dist)} p_U^{}\stensor{U}\\
  &\preceq_{\lift{\Cov^*}} \mu\frac{1 + \bss}{(1 - \sigma)(1 - \eps)} \sum_{U \in \Fcal(\Dist)} p_U^{}\stensor{U}.
  &&&&\qedhere
\end{detailedproof}
\end{proof}

\begin{proposition}
\label{prop:polyhedral-case-certificate}

Let \(\eps, \sigma, \Chernoff \in (0, 1)\).
Let \(\Dist, \Cov \subseteq \Sym{n}\) be closed convex cones such
that~\cref{eq:conic-assumptions} holds, where \(\Cov \coloneqq
\Acal(\Lp{d})\) for a linear map \(\Acal \colon \Reals^d
\to \Sym{n}\).
Assume~\cref{eq:I-in-Dist} and that
\begin{equation}
  \label{eq:central-distribution}
  \Acal^*(I) \ge \kappa \ones,
\end{equation}
for some \(\kappa \in \Reals_{++}\).
Let \(\Round[]\) be a randomized rounding algorithm for \(\Dist\).
Set \(\beta \coloneqq \Xalpha (1 - \Chernoff)(1 - \sigma) (1-\eps)\).
Let \((W, Z) \in \GWOpt_{\eps, \sigma}(\conepair)\) be such that \(W \neq 0 \neq Z\).
Let \((\bar{\rho},x)\) and \((\bar{\mu},Y)\) witness the membership
\((W, Z) \in \GWOpt_{\eps, \sigma}(\conepair)\).
There exists a randomized polynomial-time algorithm that takes
\((\bar{\rho}, x)\) and \((\bar{\mu}, Y)\) as input and outputs
a \(\beta\)-certificate \((\rho, \mu, y, U, x)\) for \((W, Z)\) with
high probability, and such that
\begin{equation*}
  \card{\supp(y)} \leq \ceil[\bigg]{
    \frac{2(\log(d) + \log(n))}{\kappa \eps \Xalpha \Chernoff^2}
  }
  \text{ almost surely}.
\end{equation*}
In~particular, \((W,Z)\) is a \(\beta\)-pairing.
\end{proposition}
\begin{proof}

Note that \(\bar{\rho}, \bar{\mu} > 0\) as \(W \neq 0 \neq Z\).
Set~\(\rho \coloneqq (1-\eps)^{-1} \bar{\rho}\) and
\(\mu \coloneqq (1/\rho) \iprod{W}{Z}\).
Note that \cref{item:cert-1} holds trivially.
We~also have \cref{item:cert-4}, since
\(\Diag(x) \succeq_{\Dist^*} W\) and
\begin{equation*}
  \rho = \frac{\bar{\rho}}{1-\eps}
  \geq\frac{1}{1-\eps} \paren[\Big]{
    (1-\eps)\iprodt{\ones}{x} +
    \eps\trace(W)}
  \geq \iprodt{\ones}{x}
\end{equation*}
as \(I \in \Dist\) by~\cref{eq:I-in-Dist} and \(\Dist
\subseteq \lift{\Cov^*}\) by~\cref{eq:lift-cov-dua-contains-dist}.

We now prove \cref{item:cert-3}.
Set \(\bar{Y} \coloneqq \bar{\mu}^{-1} Y\).
Since \((\bar{\mu}, Y)\) is feasible in~\cref{eq:FEVC-eps-gaugef} and
\(\Dist \subseteq \lift{\Cov^*}\), we have that \(\bar{Y}
\succeq_{\lift{\Cov^*}} \eps I\).
Thus \(\Acal^*(\bar{Y}) \ge \eps \Acal^*(I) \ge \eps \kappa \ones\)
by~\cref{polyhedral-conic-ineq} and~\cref{eq:central-distribution}.
\Cref{prop:polyhedral-conic-sampling} ensures that one can compute
\(\bar{y} \in \Lp{\Fcal(\Dist)}\) such that, with probability at least
\(1 - 1/n\),
\begin{equation}
  \label{eq:probabilistic-cover}
  \frac{1}{(1 - \Chernoff)}\frac{1}{\Xalpha}
  \sum_{U \in \Fcal(\Dist)} \bar{y}_U^{} \stensor{U}
  \succeq_{\lift{\Cov^*}} \frac{1}{\bar{\mu}} Y.
\end{equation}
Setting \(p \coloneqq \paren{\Xalpha(1 - \Chernoff)}^{-1} \bar{y}\)
and \(\bss \coloneqq 0\), \cref{prop:certificate-auxiliary} finishes
the proof.
\end{proof}
\begin{proof}[Proof of \cref{thm:polyhedral-algorithm}]

Set \(\tau \coloneqq 1 - \beta/\Xalpha\), and \(\sigma \coloneqq
\Chernoff \coloneqq \eps \coloneqq \tau/3\).
If we are given \(w \in \Lp{d}\) as input, nearly
solve~\cref{eq:MAXQ-eps-def} with \(W \coloneqq \Acal(w)\) and set \(z
\coloneqq \Acal^*(Y)\) and \(Z \coloneqq Y\).
If we are given \(z \in \Lp{d}\), nearly
solve~\cref{eq:nu-projected-psd} and set \(Z \coloneqq Y\).
In both cases, we obtain \((\Acal(w), Z) \in \GWOpt_{\eps,
\sigma}(\Dist, \Cov)\) such that \(\Acal^*(Z) \ge z\), as well as the
appropriate witnesses of this membership.
By definition, it suffices to obtain a \(\beta\)-certificate for
\((\Acal(w), Z)\).
\Cref{prop:polyhedral-case-certificate}
implies one can compute, in polynomial time, a
\(\hat{\beta}\)-certificate \((\rho, \mu, U, y, x)\) for \((\Acal(w),
Z)\) with \(\hat{\beta} = \Xalpha(1 - \Chernoff)(1 - \sigma)(1 -
\eps)\).
Since \(0 < \tau < 1 \le 9\), we have that
\[
  \hat{\beta}
  = \Xalpha
  \paren[\bigg]{
    1 - \frac{\tau}{3}
  }^3
  = \Xalpha\paren[\bigg]{
    1 - \tau + \dfrac{1}{3}\tau^2 - \dfrac{1}{27}\tau^3
  }
  \ge \Xalpha(1 - \tau)
  = \beta.
\]
This implies that \((\rho, \mu, U, y, x)\) is a \(\beta\)-certificate.
By \cref{prop:polyhedral-case-certificate}, we have that
\[
  \card{\supp(y)}
  \le \ceil[\bigg]{
    \frac{2}{\kappa\eps\Xalpha\Chernoff^2} (\log(d) + \log(n))
  }
  = \ceil[\bigg]{
    \frac{54}{\kappa\Xalpha\tau^3}(\log(d) + \log(n))
  }.
  \qedhere
\]
\end{proof}

\begin{proposition}
\label{prop:sdp-case-certificate}

Let \(\eps, \sigma, \Chernoff, \bss \in (0, 1)\).
Let \(\Dist, \Cov \subseteq \Sym{n}\) be such
that~\cref{eq:conic-assumptions} holds, and assume that \(\Cov
\subseteq \Psd{n}\).
Let \(\Round[]\) be a randomized rounding algorithm for \(\Dist\).
Set \(\beta \coloneqq \Xalpha (1 - \Chernoff)(1 - \sigma) (1-\eps)/(1
+ \bss)\).
Let \((W, Z) \in \GWOpt_{\eps, \sigma}(\Dist, \Cov)\) be such that \(W \neq
0 \neq Z\).
Let \((\bar{\rho},x)\) and \((\bar{\mu},Y)\) witness the membership
\((W, Z) \in \GWOpt_{\eps, \sigma}(\Dist, \Cov)\).
There exists a polynomial time algorithm that takes \((\bar{\rho},
x)\) and \((\bar{\mu}, Y)\) as input and outputs a
\(\beta\)-certificate \((\rho, \mu, y, U, x)\) with high probability.
Almost surely, we have that \(\card{\supp(y)} \in O(n/\bss^2)\) and
that the algorithm takes at most
\begin{equation*}
  \ceil[\bigg]{
    \frac{8n^2\log(2n)}{(\eps \Xalpha \Chernoff)^2}
  },
\end{equation*}
samples from \(\Round[\bar{\mu}^{-1} Y]\).
In~particular, \((W,Z)\) is a \(\beta\)-pairing.
\end{proposition}
\begin{proof}

Let \((W, Z) \in \GWOpt_{\eps, \sigma}(\Dist, \Cov)\).
Let \((\bar{\rho}, x)\) and \((\bar{\mu}, Y)\) witness the membership
\((W, Z) \in \GWOpt_{\eps, \sigma}(\Dist, \Cov)\).
Note that \(\bar{\rho}, \bar{\mu} > 0\) as \(W \neq 0 \neq Z\).
Set \(\rho \coloneqq (1 - \eps)^{-1}\bar{\rho}\) and \(\mu \coloneqq
\rho^{-1}\iprod{Z}{W}\).
Then~\cref{item:cert-1} holds trivially.
We also have \cref{item:cert-4}, since \(\Diag(x) \succeq_{\Dist^*}
W\) and
\[
  \rho
  = \frac{\bar{\rho}}{1 - \eps}
  \ge \frac{1}{1 - \eps}\paren{
    (1 - \eps)\iprodt{\ones}{x} + \eps\trace(W)
  }
  \ge \iprodt{\ones}{x},
\]
as \(I \in \Dist\) by~\cref{eq:I-in-Dist} and \(\Dist \subseteq
\Cov^*\) by~\cref{eq:conic-assumptions}.

We now prove \cref{item:cert-3}.
Set \(\bar{Y} \coloneqq \bar{\mu}^{-1}Y\).
Since \((\bar{\mu}, Y)\) is feasible in~\cref{eq:FEVC-eps-gaugef} and
\(\Dist \subseteq \lift{\Cov^*}\), we have that \(\bar{Y}
\succeq_{\lift{\Cov^*}} \eps I\).
\Cref{prop:conic-sampling} ensures one can compute, in polynomial
time, \(\bar{y} \in \Lp{\Fcal(\Dist)}\) such that
\(\iprodt{\ones}{\bar{y}} \le 1 + \bss\) and \(\card{\supp(\bar{y})}
\in O(n/\bss^2)\), and with probability at least \(1 - 1/(2n)\),
\[
  \sum_{U \in \Fcal(\Dist)} y_U^{}\stensor{U}
  \succeq_{\lift{\Cov^*}} (1 - \Chernoff)\Xalpha\bar{Y}.
\]
Assume this event holds, and set \(p \coloneqq (\Xalpha(1 -
\Chernoff))^{-1}y\).
\Cref{prop:certificate-auxiliary} finishes the proof.
\end{proof}

\begin{proposition}
\label{prop:nesterov-rounding-nlogn}

Let \(\xi,\,\gamma\) be such that \(\xi \in (0, 1]\) and \(\gamma \in
(0, 1)\).
Let \(\Round[]\) be a randomized rounding algorithm for \(\Dist\).
Let \(Z \in \Psd{n}\).
Let \((\mu, Y)\) be feasible in~\cref{eq:nu-polar-A-K-gaugef} with
\(\mu>0\).
Let \(T \ge \ceil[\Big]{\frac{2\pi}{\Chernoff^2\xi}\log(2n)}\), and let
\((X_t)_{t \in [T]}\) be {i.i.d.} random variables sampled from \(\Round[Y]\).
If
\begin{align}
  \label{eq:30}
  Y \succ 0
  \quad\text{and}\quad
  \qform{Y^{-1}}{s_U}
  \le \frac{1}{\xi\mu}
  \text{ for every }
  U \subseteq [n]
  \text{ with }
  \prob(\Round[Y] = U) > 0,
\end{align}
then
\[
\prob\paren[\bigg]{\,
  \frac{\mu}{(1 - \Chernoff)T} \sum_{t \in [T]} X_t
  \succeq \Xalpha
  Z
  } \ge 1 - \frac{1}{2n}.
\]
\end{proposition}
\begin{proof}
Set \(\bar{Y} \coloneqq \mu^{-1}Y\).
Since \(\mu>0\) and \(\diag(Y)=\mu \ones\), we have that \(\bar{Y}
\succ 0\) and \(\diag(\bar{Y}) = 1\).
Note that
\begin{equation}
  \label{eq:symmetric-sampling-lowerbound-2}
  \lambdamax(\bar{Y}^{-1/2}\stensor{U}\bar{Y}^{-1/2})
  = \mu\qform{Y^{-1}}{s_U}
    \le \mu\frac{1}{\mu\xi}
  \le \frac{1}{\xi}.
\end{equation}
Hence \(0 \preceq \Round[Y] \preceq \frac{1}{\xi} \bar{Y}\)
almost surely.
From \cref{eq:rounding-constant}, we may apply
\Cref{prop:psd-concentration} with \(\rho \coloneqq 1/\xi\) and \(\tau
\coloneqq \Nalpha = 2/\pi\) to conclude that with
\[
  T
  \ge \ceil[\bigg]{\frac{4\rho}{\Chernoff^2\tau}\log(2n)}
  = \ceil[\bigg]{\frac{4}{\Chernoff^2\xi\Nalpha}\log(2n)}
  = \ceil[\bigg]{\frac{2\pi}{\Chernoff^2\xi}\log(2n)}
\]
we have that \(\tfrac{1}{T}\sum_{t \in [T]} X_t \succeq (1 -
\Chernoff) \Xalpha \bar{Y}\) with probability at least \(1 - 1/(2n)\).
The result follows from \((\mu, Y)\) being feasible
in~\cref{eq:nu-polar-A-K-gaugef}.
\end{proof}

\begin{proof}[Proof of \cref{theorem:sdp-contained-algorithm}]

Set \(\tau \coloneqq 1 - \beta/\Xalpha\), and set \(\Chernoff
\coloneqq \sigma \coloneqq \eps \coloneqq \tau/4\).
Set \(\bss \coloneqq \tau/(4 - \tau)\), so \((1 + \bss)^{-1} = 1 -
\tau/4\).
By nearly solving either~\cref{eq:MAXQ-eps-def}
or~\cref{eq:FEVC-eps-def}, depending on whether \(W \in \Cov\) or \(Z
\in \Cov^*\) was given as input, one can compute \((\bar{\rho}, x)\)
and \((\bar{\mu}, Y)\) witnessing the membership \((W, Z) \in \GWOpt_{\eps,
  \sigma}(\Dist, \Cov)\).
\Cref{prop:sdp-case-certificate} ensures one can compute, in randomized
polynomial time, a \(\hat{\beta}\)-certificate \((\rho, \mu, U, y, x)\) for
\((W, Z)\) with high probability, where \(\hat{\beta} = (1 - \tau/4)^4\).
Since the function \(x \mapsto (1 - x/4)^4\) is convex, it
overestimates \(1 - x\), which is its best linear approximation at \(x
\coloneq 0\).
Hence
\[
  \hat{\beta}
  = \Xalpha(1 - \tau/4)^4
  \ge \Xalpha(1 - \tau)
  = \beta.
\]
Thus \((\rho, \mu, U, y, x)\) is a \(\beta\)-certificate.
Note that \(\bss\) is independent of \(n\), so from
\cref{prop:sdp-case-certificate} we can conclude that
\(\card{\supp(y)} \in O(n)\), the hidden constant depending only on
\(\beta\).
Similarly, since \(\eps\) and \(\Chernoff\) do not depend on \(n\)
either, we have that the algorithm takes at most
\[
  \ceil[\bigg]{
    \frac{8}{(\eps\Xalpha\Chernoff)^2} n^2\log(2n)
  }
  = \ceil[\bigg]{
    \frac{2048}{\Xalpha^2\tau^4} n^2\log(2n)
  }
  \in O(n^2\log(n))
\]
samples from \(\Round[]\), the hidden constant depending only on \(\beta\).

Now assume \(\Cov = \Psd{n}\).
It is immediate that~\cref{eq:conic-assumptions} holds.
Set \(\bar{Y} \coloneqq \bar{\mu}^{-1} Y\).
We have that \(\bar{Y} \succeq \eps I\), so \(\bar{Y}^{-1} \preceq
\frac{1}{\eps} I\).
Hence \(\iprodt{s_U}{\bar{Y}^{-1}s_U} \le n/\eps\) for every \(U \subseteq
[n]\), so~\cref{eq:30} holds for \(\xi \coloneqq \eps/n\).
\Cref{prop:nesterov-rounding-nlogn,prop:certificate-auxiliary} imply
that we can compute, with high probability, a
\(\hat{\beta}\)-certificate \((\rho, \mu, U, \bar{y}, x)\) for \((W,
Z)\) with
\[
  \ceil[\bigg]{
    \frac{2\pi}{\xi\Chernoff^2}\log(2n)
  }
  = \ceil[\bigg]{
    \frac{2\pi}{\eps\Chernoff^2}n \log(2n)
  }
  = \ceil[\bigg]{
    \frac{128 \pi}{\tau^3}n \log(2n)
  }
\]
samples from \(\GWrv\).
Since \(\tau\) is independent of \(n\), \Cref{prop:sparsifying}
implies that we can sparsify \(\bar{y} \in \Lp{\Fcal(\Dist)}\), which
potentially has support \(O(n \log n)\), into \(y \in
\Lp{\Fcal(\Dist)}\) with \(\card{\supp(y)} \in O(n)\).
\end{proof}

\section{Boolean 2-CSP}
\label{ssec:boolean-csp-appendix}

Let \(i,j \in [n]\).
Define the following matrices in \(\Sym{\set{0}\cup [n]}\):
\begin{align*}
    \Delta_{-i, -j}
    &\coloneqq \tfrac{1}{2}\paren{
      \oprod{(e_0 - e_i)}{(e_0 - e_j)}
      +
      \oprod{(e_0 - e_j)}{(e_0 - e_i)}
      },\\
    \Delta_{-i, +j}
    &\coloneqq \tfrac{1}{2}\paren{
      \oprod{(e_0 - e_i)}{(e_0 + e_j)}
      +
      \oprod{(e_0 + e_j)}{(e_0 - e_i)}
      },\\
    \Delta_{+i, -j}
    &\coloneqq \tfrac{1}{2}\paren{
      \oprod{(e_0 + e_i)}{(e_0 - e_j)}
      +
      \oprod{(e_0 - e_j)}{(e_0 + e_i)}
      },\\
    \Delta_{+i, +j}
    &\coloneqq \tfrac{1}{2}\paren{
      \oprod{(e_0 + e_i)}{(e_0 + e_j)}
      +
      \oprod{(e_0 + e_j)}{(e_0 + e_i)}
      }.
\end{align*}
Direct computation shows that, for every \(B \in \Reals^{\set{0,\dotsc,n}
  \times \set{0,\dotsc,n}}\),
\begin{equation*}
  \begin{aligned}
    \iprod{\Delta_{-i, -j}}{B^\transp B}
    &= \iprod{B(e_0 - e_i)}{B(e_0 - e_i)},
    \quad
    \iprod{\Delta_{-i, +j}}{B^\transp B}
    &= \iprod{B(e_0 - e_i)}{B(e_0 + e_i)}\\
    \iprod{\Delta_{+i, -j}}{B^\transp B}
    &=\iprod{B(e_0 + e_i)}{B(e_0 - e_i)},
    \quad
    \iprod{\Delta_{+i, +j}}{B^\transp B}
    &= \iprod{B(e_0 + e_i)}{B(e_0 + e_i)}.
  \end{aligned}
\end{equation*}
It is then routine to check that
\begin{equation}
  \label{eq:triangle-ineq-preserves-sign-vectors}
  \conv\setst{\stensor{U}}{U \subseteq \set{0} \cup [n]}
  \cap \Dist_{\Delta}
  =  \conv\setst{\stensor{U}}{U \subseteq \set{0} \cup [n]}
  = \CUT^{\Dist_{\Delta}}.
\end{equation}

\printbibliography

\end{document}